\pdfoutput=1
\documentclass[acmsmall,nonacm]{acmart}

\usepackage[utf8]{inputenc}
\usepackage{microtype}
\usepackage[bibliography=common]{apxproof}
\usepackage[linesnumbered,algoruled,commentsnumbered,boxruled]{algorithm2e}

\usepackage[notion,paper]{knowledge}
\usepackage{paralist}

\newtheoremrep{theorem}{Theorem}[section]

\newtheorem{corollary}[theorem]{Corollary}

\newtheoremrep{proposition}[theorem]{Proposition}
\newtheoremrep{fact}[theorem]{Fact}
\newtheoremrep{lemma}[theorem]{Lemma}
\newtheoremrep{example}[theorem]{Example}
\newtheoremrep{definition}[theorem]{Definition}

\newcommand{\NN}{\mathbb{N}}

\renewcommand\epsilon\varepsilon
\renewcommand\phi\varphi
\renewcommand\leq\leqslant
\renewcommand\geq\geqslant
\newcommand\defeq{\stackrel{\mathrm{def}}{=}}

\newcommand\rev[1]{{#1}}

\newcommand\calF{\mathcal{F}}

\newcommand\QQ{\mathbb{Q}}
\newcommand\bD{\mathbf{D}}

\newcommand{\tr}{\leq_{\mathsf{P}}}
\newcommand{\equivt}{\equiv_{\mathsf{P}}}

\newcommand\escore{\mathsf{EScore}}
\newcommand\score{\mathsf{Score}}
\newcommand\eshapley{\escore_{\cshapley}}
\newcommand\shapley{\score_{\cshapley}}
\newcommand\ebanzhaf{\escore_{\cbanzhaf}}

\newcommand\ev{\mathsf{EV}}
\newcommand\evs{\mathsf{EV}_\star}
\newcommand\ennv{\mathsf{ENV}} 
\newcommand\envss{\mathsf{ENV}_{\star,\star}}
\newcommand\mc{\mathsf{MC}}
\newcommand\env{\mathsf{ENV}}

\newcommand\vars{\mathsf{Vars}}
\newcommand\out{\mathsf{out}} 
\newcommand\cshapley{c_{\mathrm{Shapley}}}
\newcommand\cbanzhaf{c_{\mathrm{Banzhaf}}}
\newcommand\ar{\mathsf{ar}}
\newcommand\const{\mathsf{Const}}
\newcommand{\pqe}{\mathrm{PQE}}

\newcommand{\hrulealg}[0]{\vspace{1mm} \hrule \vspace{1mm}}

\newcommand{\0}{\phantom{0}}

\newcommand{\phiex}{\phi_{\mathrm{ex}}}

\knowledge{tractable}{notion}
\knowledge{reasonable}{notion}
\knowledge{closed under conditioning}{notion}
\knowledge{closed under conjunctions (resp., disjunctions) with fresh variables}
  [closed under either conjunctions or disjunctions with fresh variables]{notion}
\knowledge{coefficient function}[coefficients|coefficient functions]{notion}
\knowledge{notion}
 | deterministic and decomposable (d-D) Boolean circuit
 | d-D circuit
 | d-Ds
 | d-D
 | d-D circuits
\knowledge{tight}{notion}

\allowdisplaybreaks

\author{Pratik Karmakar}
\orcid{0009-0008-1111-8801}
\affiliation{\institution{National University of Singapore}\department{School of Computing}\city{Singapore}\country{Singapore}}
\affiliation{\institution{CNRS@CREATE Ltd}\city{Singapore}\country{Singapore}}
\email{pratik.karmakar@u.nus.edu}

\author{Mikaël Monet}
\affiliation{\institution{Inria}\city{Villeneuve d'Ascq}\country{France}}
\affiliation{\institution{CRIStAL, Université de Lille, CNRS}\city{Villeneuve d'Ascq}\country{France}}
\orcid{0000-0002-6158-4607}
\email{mikael.monet@inria.fr}

\author{Pierre Senellart}
\affiliation{\institution{DI ENS, ENS, PSL University, CNRS}\city{Paris}\country{France}}
\affiliation{\institution{Inria}\city{Paris}\country{France}}
\affiliation{\institution{IUF}\city{Paris}\country{France}}
\additionalaffiliation{\institution{CNRS@CREATE Ltd \& IPAL, CNRS}\city{Singapore}\country{Singapore}}
\orcid{0000-0002-7909-5369}
\email{pierre@senellart.com}

\author{Stéphane Bressan}
\affiliation{\institution{National University of Singapore}\department{School of Computing}\city{Singapore}\country{Singapore}}
\authornotemark[1]{
\orcid{0000-0001-5536-3296}
\email{steph@nus.edu.sg}

\authorsaddresses{Authors’ addresses:
  Pratik Karmakar, NUS, School of Computing, Singapore, pratik.karmakar@u.nus.edu;
  Mikaël Monet, Inria, Villeneuve d’Ascq, France, mikael.monet@inria.fr;
  Pierre Senellart, DI ENS, ENS, PSL University, Paris, France, pierre@senellart.com;
  Stéphane Bressan, NUS, School of Computing, Singapore, steph@nus.edu.sg.}

\ccsdesc[500]{Mathematics of computing~Probabilistic representations}
\ccsdesc[300]{Information systems~Relational database model}
\ccsdesc[300]{Theory of computation~Algorithmic game theory and mechanism design}

\keywords{Shapley value, Banzhaf value, probabilistic database,
provenance, knowledge compilation, d-D circuit}

\begin{document}

\title{Expected Shapley-Like Scores of Boolean Functions:
Complexity~and~Applications~to~Probabilistic~Databases}

\begin{abstract}
  Shapley values, originating in game theory and increasingly prominent in
explainable AI, have been proposed to assess the contribution of facts in
query answering over databases, along with other similar power indices
such as Banzhaf values. In this work we adapt these Shapley-like scores
to probabilistic settings, the objective being to compute their expected
value. We show that the computations of expected Shapley values and of
the expected values of Boolean functions are interreducible in polynomial
time, thus obtaining the same tractability landscape. We investigate the
specific tractable case where Boolean functions are represented as
deterministic decomposable circuits, designing a polynomial-time
algorithm for this setting. We present applications to probabilistic
databases through database provenance, and an effective implementation of
this algorithm within the ProvSQL system, which experimentally validates
its feasibility over a standard benchmark.

\end{abstract}

\maketitle

\section{Introduction}
\label{sec:intro}
The \emph{Shapley value} is a popular notion from cooperative game theory, introduced by Lloyd Shapley~\cite{shapley1953value}.
Its idea is to ``fairly'' distribute the rewards of a game among the
players.
The Banzhaf power index~\cite{banzhaf1964weighted}, another power distribution index with different
weights, also plays an important role in voting theory. These are two
instances of power indices for coalitions, which also include the
Johnston~\cite{johnston1977national,johnston1978measurement},
Deegan--Packel~\cite{deegan1978new}, and Holler--Packel
indices~\cite{holler1983power}, see~\cite{laurelle1999choice} for a survey.
Shapley and Banzhaf values, in particular, have
found recent applications in explainable machine
learning~\cite{KarczmarzM0SW22,van2022tractability} and valuation of data
inputs in data
management~\cite{deutch2022computing,abramovich2023banzhaf}.

In this work, we revisit the computation of such values (which we call
\emph{Shapley-like} values or scores) in a setting where data is uncertain.
\rev{A long line
of work
\cite{owen1972multilinear,weber1988probabilistic,laruelle2008potential,carreras2015multinomial,carreras2015coalitional,koster2017prediction,borkotokey2023expected}
has studied various ways in which scoring mechanisms can be
incorporated in a probabilistic setting.
More specifically, the authors of \cite{borkotokey2023expected}
propose the notion of \emph{expected Shapley value} and show that this is the only scoring mechanism
satisfying a certain set of axioms in a probabilistic setting. Motivated by this, we extend their notion
to other scores (e.g., Banzhaf), and study the complexity of the associated computational problems.}

Our objective is then
to investigate the tractability of expected Shapley-like value
computations for Boolean functions,
having in mind the potential application of
computation of expected Shapley-like values of facts for a query over
probabilistic databases.
In
particular, some results have been obtained in the literature that
reduces the complexity of (non-probabilistic) Shapley-value computation to and from the computation of
the model count of a Boolean function (or to the computation of the
probability of a query in probabilistic databases)
under some technical conditions
\cite{deutch2022computing,kara2023shapley}; we aim at understanding this
connection better by investigating whether \emph{expected} Shapley(-like) value
computation, which combines the computation of a power index and a
probabilistic setting, is harder than each of these aspects taken in
isolation.

We provide the following contributions. First (in
Section~\ref{sec:prelims}), we formally introduce the notion of
Shapley-like scores and of the expected value of such scores on Boolean
functions whose variables are assigned independent probabilities.
In Section~\ref{sec:equiv}, we investigate the connection between the
computation of expected Shapley-like scores and the computation of the
expected value of a Boolean function. In particular, we show a very
general result
(Corollary~\ref{cor:equiv}) that expected Shapley
value computation is interreducible in polynomial time to the expected
value computation problem over any class of Boolean functions for which it
is possible to compute its value over the empty set in polynomial time;
we also
obtain a similar result (Corollary~\ref{cor:equiv_banzhaf}) for the computation of
expected Banzhaf values.
We then assume in Section \ref{sec:dD} that we have a tractable
representation of a Boolean function as a decomposable and deterministic
circuit; in this case, we show a concrete polynomial-time algorithm for
Shapley-like score computation (Algorithm~\ref{alg:dD}) and some
simplifications thereof for specific settings. We then apply in
Section~\ref{sec:pdbs} these results to the case of probabilistic
databases, showing (Corollary~\ref{cor:pqe-escore}) that expected Shapley
value computation is interreducible in polynomial time to probabilistic
query evaluation. In Section~\ref{sec:exp} we show through an
experimental evaluation that the algorithms
proposed in this paper are indeed feasible in practical scenarios. Before
concluding the paper, we discuss related work in
Section~\ref{sec:relwork}.

For space reason, some proofs are relegated to the
\hyperref[appendix]{appendix}.

\section{Preliminaries}
\label{sec:prelims}
For~$n\in \NN$ we write~$[n]\defeq\{0,\ldots,n\}$. We
denote by $\mathsf{P}$ the class of problems solvable in polynomial
time. For a set~$V$, we denote by~$2^V$ its powerset.

\paragraph*{Boolean functions.}
A \emph{Boolean function over a finite set of variables}~$V$ is a
mapping~$\phi:2^V \to \{0,1\}$. To
talk about the complexity of problems over a class of Boolean
functions, one must first specify how functions are specified as input. By
a \emph{class of Boolean functions}, we then mean a class of
\emph{representations} of Boolean functions; for instance, truth tables,
decision trees, Boolean circuits, and so on,
with any sensible encoding. In particular, we consider that the size of~$V$
is always provided in unary as part of the input.
Let~$\phi:2^V \to \{0,1\}$ and~$x\in V$. We denote
by~$\phi_{+x}$ (resp.,~$\phi_{-x}$) the Boolean function on~$V\setminus \{x\}$ that
maps~$Z\subseteq V\setminus \{x\}$ to $\phi(Z\cup \{x\})$
(resp., to $\phi(Z)$).

\paragraph*{Expected value.}
For \rev{each} $x\in V$, let $p_x \in [0,1]$ be a probability value. For~$V'
\subseteq V$ and $Z\subseteq V'$, define $\Pi_{V'}(Z)$ to be the
probability of $Z$ being drawn from~$V'$ under the assumption that every
$x\in V'$ is chosen independently with probability~$p_x$. Formally:
\(\Pi_{V'}(Z) \defeq \bigg(\prod_{x\in Z} p_x \bigg)
\times \left(\prod_{x \in V' \setminus Z} (1-p_x)\right).\)
The~$p$-values do not appear in the notation $\Pi_{V'}(Z)$: this is to
simplify notation. For~$\phi:2^V \to \{0,1\}$, define then the \emph{expected
value of~$\phi$} as $\ev(\phi) \defeq \sum_{Z\subseteq V} \Pi_V(Z)\phi(Z)$.
Note that this is simply the probability of $\phi$ being true.
We then define the corresponding problem for a class of Boolean
functions~$\calF$.

\medskip
\begin{center}
\fbox{\begin{tabular}{lp{10cm}}
    \small{PROBLEM} : & $\ev(\calF)$ \hspace{1cm}(\emph{Expected Value})
\\
{\small INPUT} : & A Boolean function $\phi \in \calF$ over variables~$V$ and
probability values~$p_x$ for each $x\in V$
\\
{\small OUTPUT} : & The quantity $\ev(\phi)$
\end{tabular}}
\end{center}
Here, we consider as usual that the probabilities values are rational
numbers $\frac{p}{q}$ for~$(p,q)\in \NN\times \rev{(\NN \setminus \{0\})}$, provided as ordered pairs $(p,q)$
where $p$ and $q$ themselves are encoded in binary.

\begin{example}\label{exa:ev}
  \rev{Consider the Boolean function in disjunctive normal form:
  $\phiex=(A\land
  a)\lor (C\land c)$ on variables $\{A,a,C,c\}$. Assume the probabilities
  for these variables given in the second column of
  Table~\ref{tab:shap-eshap}.
  Then
  its expected value can be computed (since $\phiex$ is read-once) as:
  $\ev(\phiex)=1-(1-0.4\times 0.5)\times(1-0.6\times 0.8)=1-0.8\times
0.52=0.584$.}
\end{example}

\paragraph*{Shapley-like scores.}
Let~$c:\NN \times \NN \to \QQ$ be a function, that we call the
\intro{coefficient function}, and let~$\phi:2^V \to \{0,1\}$ and~$x\in V$.
Define the \emph{Shapley-like score with coefficients~$c$ of $x$ in $V$\! with
respect to~$\phi$}, or simply \emph{score} when clear from context, by
\[\score_c(\phi,V,x) \defeq \sum_{E \subseteq V\setminus \{x\}}
    c(|V|,|E|)\times \big[\phi(E\cup \{x\}) - \phi(E)\big].
  \]
  \begin{example}
    Let $\cshapley(k,\ell) \defeq
    \frac{\ell!(k-l-1)!}{k!}=\binom{k-1}{l}^{-1}k^{-1}$ and
    $\cbanzhaf(k,\ell)$ $\defeq 1$. Then, by definition, $\score_{\cshapley}(\phi,V,x)$
    (resp., $\score_{\cbanzhaf}(\phi,V,x)$) is the usual Shapley (resp., Banzhaf)
    value, with set of players~$V$ and wealth function~$\phi$. The
    Penrose--Banzhaf power~\cite{kirsch2010power}, a
    normalization of Banzhaf
  values, can also be defined by \kl{coefficients} $(k,\ell)\mapsto
  2^{-k+1}$.
  \end{example}
  For each fixed \kl{coefficient function}~$c$ and class of Boolean functions~$\calF$, we
define the corresponding computational problem.
\medskip
\begin{center}
\fbox{\begin{tabular}{lp{10cm}}
    \small{PROBLEM} : & $\score_c(\calF)$
\\
{\small INPUT} : & A Boolean function $\phi \in \calF$ over variables~$V$, a variable~$x\in V$
\\
{\small OUTPUT} : & The quantity $\score_c(\phi,V,x)$
\end{tabular}}
\end{center}

\begin{table}
    \caption{\rev{(Expected) Shapley values for
    $\phi_{\mathrm{ex}}=(A\land a)\lor (C\land c)$}}
  \centering\small
    \rev{\begin{tabular}{cccc}
      \toprule
      $x\in V$ & $p_x$ & $\score_{\cshapley}(\phi_{\mathrm{ex}},V,x)$&
      $\escore_{\cshapley}(\phi_{\mathrm{ex}}, x)$\\
        \midrule
        A & 0.4 & 0.25 & 0.076 \\
        a & 0.5 & 0.25 & 0.076 \\
        C & 0.6 & 0.25 & 0.216 \\
        c & 0.8 & 0.25 & 0.216 \\
        \bottomrule
    \end{tabular}}
  \label{tab:shap-eshap}
\end{table}

\paragraph*{Expected Shapley-like scores.}
We now introduce the probabilistic variant of Shapley-like scores, which is
our main object of study.
\begin{definition}
  \label{def:eshapley}
  Let~$c$ be a \kl{coefficient function}, $\phi:2^V \to \{0,1\}$ a~Boolean function
  over variables~$V$, probability values $p_y$ for $y\in V$, and~$x\in V$.
  Define the \emph{expected score of~$x$ for~$\phi$} as:
  \[\escore_c(\phi,x) \defeq \sum_{\substack{Z\subseteq V\\ x\in Z}}
    \left(\Pi_V(Z)\times \score_c(\phi,Z,x)\right),
  \]
  where in $\score_c(\phi,Z,x)$ we see~$\phi$ as a function from~$2^Z$ to~$\{0,1\}$.
\end{definition}
In words, this is the expected value of the corresponding score, when players are
independently selected to be part of the cooperative game. Notice that
subsets~$Z\subseteq V$ not containing~$x$ are not summed over: this is because in
this case~$x$ is not a player of the selected game and we thus declare its
“contribution” to be null.
\rev{We point out that the notion of expected Shapley value has been defined and studied in \cite{borkotokey2023expected}
in terms of the properties that it satisfies. The authors show that this is the only value satisfying some natural sets of axioms, making it
a natural candidate for score attribution in a probabilistic setting.}
This definition is also strongly motivated by its
applications to probabilistic databases (cf.
Section~\ref{sec:pdbs}). We then define the corresponding computational problem, for
a fixed~$c$ and~$\calF$.
\medskip
\begin{center}
\fbox{\begin{tabular}{lp{10cm}}
    \small{PROBLEM} : & $\escore_c(\calF)$ \hspace{1cm} (\emph{Expected Score})
  \\
{\small INPUT} : & A Boolean function $\phi \in \calF$ over variables~$V$,
probability values~$p_y$ for each $y\in V$, a variable~$x\in V$
\\
{\small OUTPUT} : & The quantity $\escore_c(\phi,x)$
\end{tabular}}
\end{center}

\begin{example}
  \rev{Continuing from Example~\ref{exa:ev}, we can compute the Shapley value
  and the expected Shapley value of variables appearing in $\phiex$. By symmetry,
  it is easy to see that all variables have the same Shapley value; on
  the other hand the expected Shapley values depend on the probability of
  each clause to be true. See Table~\ref{tab:shap-eshap} for all
  (expected) Shapley values. If other variables were in $V$, their
  Shapley and expected Shapley values w.r.t.\ $\phiex$ would be $0$, regardless
  of their probabilities, since they do not play a role in the
satisfaction of~$\phiex$.}
\end{example}

\paragraph*{Reductions.}
For two computational problems~$A$ and~$B$,
we write $A \tr B$ to assert the existence of a polynomial-time Turing
reduction from~$A$ to~$B$, i.e., a polynomial-time reduction that is allowed to use~$B$ as an
oracle. We write~$A \equivt B$ when $A \tr B$ and $B \tr A$, meaning that the
problems are equivalent under such reductions.
Using this notation we can state
a first trivial fact:
\begin{fact}
  \label{fact:eshapley-to-shapley}
We have $\score_c(\calF) \tr \escore_c(\calF)$ for any \kl{coefficient
function}~$c$ and class of Boolean functions~$\calF$.
\end{fact}
This is simply because $\escore_c(\phi,x) = \score_c(\phi,V,x)$ when $p_y=1$
for all~$y\in V$, so that our probabilistic variants of such scores are proper
generalizations of the non-probabilistic ones.

\section{Equivalence with Expected Values}
\label{sec:equiv}
\begin{toappendix}
  \label{appendix}
\end{toappendix}

In this section we link the complexity of computing expected Shapley-like
scores with that of computing expected values. The point is that $\ev(\calF)$ is
a central problem that has already been studied in depth for most meaningful
classes of Boolean functions, with classes for which that problem is in
$\mathsf{P}$ while the general problem is $\mathsf{\#P}$-hard.
In a sense then, if we can show for some problem
$A$ that $A \equivt \ev(\calF)$, this settles the complexity of $A$.
We start in Section~\ref{subsec:ev-to-score} by the direction most
interesting for efficient algorithms:
\rev{from expected scores to expected values}.
We show that this is always possible,
under the assumption that the \kl{coefficient function} is computable in
polynomial time. We then give results for the other direction in
Section~\ref{subsec:escore-to-ev}, where the picture is more complex.

\subsection{\rev{From Expected Scores to Expected Values}}
\label{subsec:ev-to-score}

Let us call a \kl{coefficient function}~$c$ \intro{tractable} if $c(k,\ell)$ can
be computed in $\mathsf{P}$ when $k$ and $\ell$ are given in unary as input. It is
easy to see that $\cbanzhaf$ and its normalized version are
\kl{tractable}.
This is also the case of~$\cshapley$, using the fact that the
$\binom{k}{\ell}$ binomial
coefficient can be computed in time $O(k\times\ell)$ by dynamic programming
(assuming arguments in unary). Under this
assumption, we show that computing expected Shapley-like scores reduces in polynomial
time to computing expected values.

\begin{theorem}
  \label{thm:ev-to-escore}
We have~$\escore_c(\calF) \tr \ev(\calF)$
for any \kl{tractable} \kl{coefficient function}~$c$ and any class
$\calF$\! of Boolean functions.
\end{theorem}

We obtain for instance that~$\escore_c(\calF)$ is in $\mathsf{P}$ for decision trees,
ordered binary decision diagrams (OBDDs), deterministic and decomposable
Boolean circuits, Boolean circuits of bounded
treewidth~\cite{amarilli2020connecting}, and so on, since~$\ev(\calF)$ is
in $\mathsf{P}$
for these classes. By Fact~\ref{fact:eshapley-to-shapley},
this also recovers the results from \cite{deutch2022computing,kara2023shapley,abramovich2023banzhaf}
that (non-expected) Shapley and Banzhaf scores are in $\mathsf{P}$ for the tractable classes
that they consider.

We prove Theorem~\ref{thm:ev-to-escore} in the remaining of this section. To do
so, we first define two intermediate problems.
\medskip
\begin{center}
\fbox{\begin{tabular}{lp{10cm}}
    \small{PROBLEM} : & $\evs(\calF)$ \hspace{.5cm}(\emph{Expected Value of Fixed Size})
\\
{\small INPUT} : & A Boolean function $\phi \in \calF$ over variables~$V$,
probabilities $p_x$ for each $x\in V$, and~$k\in [|V|]$
\\
{\small OUTPUT} : & The quantity $\ev_k(\phi) \defeq \sum_{\substack{Z\subseteq V\\|Z|=k}} \Pi_V(Z)\phi(Z)$
\end{tabular}}
\end{center}
\medskip
\begin{center}
\fbox{\begin{tabular}{lp{10cm}}
    \small{PROBLEM} : & $\envss(\calF)$ \hspace{.1cm}(\emph{Expected Nested Value of Fixed Sizes})
\\
{\small INPUT} : & A Boolean function $\phi \in \calF$ over $V$,
probabilities $p_x$ for each $x\in V$, and integers~$k,\ell \in [|V|]$
\\
{\small OUTPUT} : & The quantity $\ennv_{k,\ell}(\phi) \defeq
\sum_{\substack{Z\subseteq V\\|Z|=k}} \Pi_V(Z) \sum_{\substack{E\subseteq
Z\\|E|=\ell}}\phi(E)$
\end{tabular}}
\end{center}
Notice that~$\ennv_{k,\ell}(\phi)= 0$ when~$k < \ell$. Also, observe
that~$\ev(\phi) = \sum_{k=0}^{|V|} \ev_{k}(\phi)$ and that $\ev_k(\phi) =
\ennv_{k,k}(\phi)$, so that~$\ev(\calF) \tr \evs(\calF) \tr \envss(\calF)$
for any~$\calF$.

We now prove the chain of reductions
$\escore_c(\calF) \tr \envss(\calF) \tr \evs(\calF) \tr \ev(\calF)$, in this
order, which implies Theorem~\ref{thm:ev-to-escore} indeed.

\begin{lemma}
\label{lem:envss-to-escore}
We have $\escore_c(\calF) \tr
\envss(\calF)$ for any \kl{tractable} coefficient function $c$ and any class of Boolean functions~$\calF$.
\end{lemma}
\begin{proof}
Let~$\phi:2^V \to \{0,1\}$ in $\calF$, probability values~$p_y$ for~$y\in V$, and~$x\in V$.
We wish to compute~$\escore_c(\phi,x)$.
Observe that $\escore_c(\phi,x) = A - B$, where
\[
  A = \sum_{\substack{Z\subseteq V\\ x\in Z}} \Pi_V(Z) \sum_{E \subseteq
  Z\setminus \{x\}}
    c(|Z|,|E|) \phi(E\cup \{x\})\qquad
  B = \sum_{\substack{Z\subseteq V\\ x\in Z}} \Pi_V(Z) \sum_{E \subseteq
  Z\setminus \{x\}}
c(|Z|,|E|) \phi(E).\]
Let us focus on~$A$. Letting~$V' \defeq V\setminus \{x\}$, notice that these are
the variables over which~$\phi_{+x}$ is defined. Letting~$n\defeq |V'|$, we have
\begin{align*}
  A &= \sum_{\substack{Z\subseteq V\\ x\in Z}} \Pi_V(Z) \sum_{E \subseteq
  Z\setminus \{x\}}
  c(|Z|,|E|) \phi_{+x}(E) = p_x \sum_{Z\subseteq V'} \Pi_{V'}(Z) \sum_{E \subseteq Z}
  c(|Z|+1,|E|) \phi_{+x}(E)\\
    &= p_x \sum_{Z\subseteq V'}  \sum_{E \subseteq Z}
  c(|Z|+1,|E|) \Pi_{V'}(Z) \phi_{+x}(E) = p_x \sum_{k=0}^n \sum_{\substack{Z\subseteq V'\\|Z|=k}} \sum_{\ell=0}^k  \sum_{\substack{E \subseteq Z\\|E|=\ell}}
  c(k+1,\ell) \Pi_{V'}(Z) \phi_{+x}(E)\\
    &= p_x \sum_{k=0}^n \sum_{\ell=0}^k c(k+1,\ell) \sum_{\substack{Z\subseteq V'\\|Z|=k}}   \sum_{\substack{E \subseteq Z\\|E|=\ell}}
   \Pi_{V'}(Z) \phi_{+x}(E) = p_x \sum_{k=0}^n \sum_{\ell=0}^k c(k+1,\ell) \ennv_{k,\ell}(\phi_{+x}).\\
\end{align*}
We can do exactly the same for~$B$ (replacing~$\phi_{+x}$ by
$\phi_{-x}$), after which we obtain:
\begin{equation}
  \escore_c(\phi,x) = p_x \sum_{k=0}^n \sum_{\ell=0}^k c(k+1,\ell)
  \big(\ennv_{k,\ell}(\phi_{+x}) - \ennv_{k,\ell}(\phi_{-x})\big).
  \label{eq:envss-to-escore}
\end{equation}

We can compute the \kl{coefficients} $c(k+1,\ell)$ in $\mathsf{P}$ because $c$ is
\kl{tractable}. Therefore, all that is left to show is that we can compute in
$\mathsf{P}$ all the values~$\ennv_{k,\ell}(\phi_{+x})$ and
$\ennv_{k,\ell}(\phi_{-x})$ for~$k,\ell \in [n]$. We point out that this is not
obvious, because~$\calF$ might not be \kl{closed under conditioning}, and
unfortunately setting $p_x$ to $0$ or $1$ is not enough to directly give us the values
we want.
In \cite{kara2023shapley}, this annoying subtlety is handled by using the
closure under OR-substitutions of the class~$\calF$ (see the proof of their
Lemma~3.2). In our case, we will overcome this problem by using the fact that
we can freely choose the probabilities.

Let~$z \in [0,1]$, and for~$y\in V'=V\setminus \{x\}$ define~$p^z_y \defeq p_y$,
and $p^z_x = z$. Define~$\Pi^z$ and~$\envss^z(\phi)$ as expected. We claim
that the following equation holds, for~$i,j \in [n+1]$:
\begin{align}
  \label{eq:simulate-conditioning}
  \env^z_{i,j}(\phi) = z \big[\ennv_{i-1,j}(\phi_{-x}) + \ennv_{i-1,j-1}(\phi_{+x})\big] + (1-z) \ennv_{i,j}(\phi_{-x}),
\end{align}
where we extended the definition of $\envss(\phi_{+x})$ and $\envss(\phi_{-x})$
to have value zero for out-of-bound $(i,j)$-indices. Before proving this claim,
let us explain why this allows us to conclude. Indeed, we can then use the oracle
to~$\envss(\calF)$ with $z=0$ to compute all the values
$\ennv_{k,\ell}(\phi_{-x})$. Once these are known, we use the oracle again but
this time with~$z=1$, and thanks to Equation~\eqref{eq:simulate-conditioning}
again we can recover all the values~$\ennv_{k,\ell}(\phi_{+x})$.

Therefore, all that is left to do is prove Equation~\eqref{eq:simulate-conditioning}.
We have:
\[
  \env^z_{i,j}(\phi) \defeq \sum_{\substack{Z\subseteq V\\|Z|=i}} \Pi_V(Z) \sum_{\substack{E\subseteq Z\\|E|=j}}\phi(E)
   = \underbrace{\sum_{\substack{Z\subseteq V\\|Z|=i\\x\notin Z}}
   \Pi_V(Z) \sum_{\substack{E\subseteq Z\\|E|=j}}\phi(E)}_\alpha
   \:+\: \underbrace{\sum_{\substack{Z\subseteq V\\|Z|=i\\x\in Z}} \Pi_V(Z)
   \sum_{\substack{E\subseteq Z\\|E|=j}}\phi(E)}_\beta.
\]
With $\alpha$ and~$\beta$ the two terms defined above, we have:
\begin{align*}
  \alpha &= (1-z) \sum_{\substack{Z\subseteq V'\\|Z|=i}} \Pi_{V'}(Z) \sum_{\substack{E\subseteq Z\\|E|=j}}\phi(E)
    = (1-z) \sum_{\substack{Z\subseteq V'\\|Z|=i}} \Pi_{V'}(Z) \sum_{\substack{E\subseteq Z\\|E|=j}}\phi_{-x}(E)
    = (1-z) \ennv_{i,j}(\phi_{-x}).
\end{align*}

Let us now inspect~$\beta$.
\[
  \beta = z \sum_{\substack{Z\subseteq V'\\|Z|=i-1}} \Pi_{V'}(Z) \sum_{\substack{E\subseteq Z\cup \{x\}\\|E|=j}}\phi(E)
    = \underbrace{z \sum_{\substack{Z\subseteq V'\\|Z|=i-1}} \Pi_{V'}(Z)
    \sum_{\substack{E\subseteq Z\cup \{x\}\\|E|=j\\x\notin
  E}}\phi(E)}_{\alpha'}
    \:+\: \underbrace{z \sum_{\substack{Z\subseteq V'\\|Z|=i-1}}
      \Pi_{V'}(Z) \sum_{\substack{E\subseteq Z\cup \{x\}\\|E|=j\\x\in
    E}}\phi(E)}_{\beta'}.
\]
Again with $\alpha'$ and~$\beta'$ two terms above, we have:
\[
  \alpha' = z \sum_{\substack{Z\subseteq V'\\|Z|=i-1}} \Pi_{V'}(Z) \sum_{\substack{E\subseteq Z\\|E|=j}}\phi_{-x}(E)
     = z \times \ennv_{i-1,j}(\phi_{-x}),
   \]
and
\[
  \beta' = z \sum_{\substack{Z\subseteq V'\\|Z|=i-1}} \Pi_{V'}(Z) \sum_{\substack{E\subseteq Z\\|E|=j-1}}\phi(E\cup \{x\})
     = z \sum_{\substack{Z\subseteq V'\\|Z|=i-1}} \Pi_{V'}(Z) \sum_{\substack{E\subseteq Z\\|E|=j-1}}\phi_{+x}(E)
     = z \times \ennv_{i-1,j-1}(\phi_{+x}).
   \]

   Putting it all together, we indeed obtain Equation~\eqref{eq:simulate-conditioning}, thus concluding the proof.
\end{proof}

The following lemma contains the most technical part of the proof of
Theorem~\ref{thm:ev-to-escore}. It is proved using polynomial interpolation
with carefully crafted probability values.

\begin{lemma}
  \label{lem:evs-to-envss}
We have $\envss(\calF) \tr \evs(\calF)$ for any~$\calF$.
\end{lemma}
\begin{proof}
  Let $\phi\in \calF$ over variables~$V$, probability values~$p_x$ for each
  $x\in V$, and~$k,\ell \in [|V|]$.  Let~$n\defeq
  |V|$. Our goal is to compute~$\ennv_{k,\ell}(\phi)$. We will in fact
  use polynomial interpolation to compute \emph{all} the values~$\ennv_{j,\ell}(\phi)$
  for~$j \in [n]$, and return $\ennv_{k,\ell}(\phi)$.

  Let~$z_0,\ldots,z_{n}$ be~$n + 1$ distinct positive values in~$\QQ$.
  For~$i\in [n]$ and~$x\in V$, define~$c^{z_i}_x \defeq
  2{z_i}p_x + 1 - p_x$ and $p^{z_i}_x \defeq \frac{z_i p_x}{c^{z_i}_x}$,
  and define~$\Pi^{z_i}$ and~$\ev^{z_i}(\phi)$ as expected. Notice that these
  are all valid probability mappings, i.e., all values~$p^{z_i}_x$ are well-defined
  and between $0$ and $1$, and observe that~$1-p^{z_i}_x = \frac{({z_i}p_x) +
  (1-p_x)}{c^{z_i}_x}$. Define further~$C_{z_i} \defeq \prod_{x\in V} c^{z_i}_x$.
  Then:
  \begin{align*}
    \ev^{z_i}_\ell(\phi) &=\sum_{\substack{E\subseteq V\\|E|=\ell}} \Pi^{z_i}_V(E)\phi(E)
     = \sum_{\substack{E\subseteq V\\|E|=\ell}} \phi(E) \prod_{x \in E} p^{z_i}_x \prod_{x \in V\setminus E} (1-p^{z_i}_x)\\
     &= \frac{1}{C_{z_i}}\sum_{\substack{E\subseteq V\\|E|=\ell}} \phi(E) \prod_{x \in E} {z_i}p_x \prod_{x \in V\setminus E} [({z_i}p_x) + (1-p_x)].
  \end{align*}
  Next we develop the innermost product as it is parenthesized and distribute the $\prod_{x \in E} {z_i}p_x$ term, obtaining:
  \begin{align}
    \ev^{z_i}_\ell(\phi) &= \frac{1}{C_{z_i}}\sum_{\substack{E\subseteq V\\|E|=\ell}} \phi(E) \sum_{E\subseteq Z \subseteq V} \prod_{x \in Z} {z_i}p_x \prod_{x \in V\setminus Z} (1-p_x)\nonumber\\
     &= \frac{1}{C_{z_i}}\sum_{\substack{E\subseteq V\\|E|=\ell}} \phi(E) \sum_{j=0}^n \sum_{\substack{E\subseteq Z \subseteq V\\|Z|=j}} \prod_{x \in Z} {z_i}p_x \prod_{x \in V\setminus Z} (1-p_x)\nonumber\\
     &= \frac{1}{C_{z_i}}\sum_{\substack{E\subseteq V\\|E|=\ell}} \phi(E) \sum_{j=0}^n {z_i}^j\sum_{\substack{E\subseteq Z \subseteq V\\|Z|=j}} \Pi_V(Z)\nonumber\\
     &= \frac{1}{C_{z_i}}\sum_{j=0}^n {z_i}^j \sum_{\substack{E\subseteq V\\|E|=\ell}} \phi(E) \sum_{\substack{E\subseteq Z \subseteq V\\|Z|=j}} \Pi_V(Z).\nonumber\\
     &= \frac{1}{C_{z_i}}\sum_{j=0}^{n} {z_i}^{j} \ennv_{j,\ell}(\phi)\label{eq:vdm1},
  \end{align}
  where in the last equality we have inverted the two innermost sums.
  Using the oracle to~$\evs(\calF)$, we compute $c_i \defeq C_{z_i} \times \ev^{z_i}_\ell(\phi)$ for~$i\in
  [n]$ in polynomial time. By
  Equation~\eqref{eq:vdm1}, this gives us a system of linear equations~$A X =
  C$,
  with $C_i \defeq c_i$, $X_j \defeq \ennv_{j,\ell}(\phi)$ and~$A_{ij}
  \defeq {z_i}^j$. We see that~$A$ is a
  non-singular Vandermonde matrix,
  so we can in polynomial time recover all the
  values~$X_j$ and return~$\ennv_{k,\ell}(\phi)$, as promised. This concludes the proof.
\end{proof}

We can finally state the last step of the proof of
Theorem~\ref{thm:ev-to-escore}, again proved using
polynomial interpolation.

\begin{lemma}
  \label{lem:ev-to-evs}
 We have $\evs(\calF) \tr \ev(\calF)$ for any~$\calF$.
\end{lemma}
\begin{proof}
  Let $\phi\in \calF$ over variables~$V$, probability values~$p_x$ for each
  $x\in V$, and~$k\in [|V|]$. Let~$n\defeq |V|$. We wish to
  compute~$\ev_k(\phi)$. We use again polynomial interpolation to compute
  all the values~$\ev_j(\phi)$ for~$j\in [n]$ and return~$\ev_k(\phi)$.

  Let~$z_0,\ldots,z_{n}$ be~$n+1$ distinct positive values in~$\QQ$.
  For~$i\in [n]$ and~$x\in V$, define~$c^{z_i}_x \defeq 1-p_x + z_i p_x$, define~$p^{z_i}_x \defeq
  \frac{z_i p_x}{c^{z_i}_x}$, and define~$\Pi^{z_i}$ and~$\ev^{z_i}(\phi)$ as expected.
  Again, these are all valid probability mappings, and observe that this time~$1-p^{z_i}_x = \frac{1 - p_x}{c^{z_i}_x}$.
  Defining as before~$C_{z_i} \defeq \prod_{x\in V} c^{z_i}_x$,
  it is this time much easier to derive the equality
  $\ev^{z_i}(\phi) = \frac{1}{C_{z_i}} \sum_{j=0}^n {z_i}^j \ev_j(\phi)$:
   \begin{align*}
     \ev^{z_i}(\phi) &\defeq \sum_{Z\subseteq V} \Pi^{z_i}_V(Z)\phi(Z)\nonumber = \sum_{j=0}^n \sum_{\substack{Z \subseteq V\\|Z|=j}} \Pi^{z_i}_V(Z)\phi(Z)\nonumber\\
                 &= \sum_{j=0}^n \sum_{\substack{Z \subseteq V\\|Z|=j}} \phi(Z) \prod_{x \in Z} p^{z_i}_x \prod_{x \in V\setminus Z} (1-p^{z_i}_x)\nonumber = \frac{1}{C_{z_i}} \sum_{j=0}^n \sum_{\substack{Z \subseteq V\\|Z|=j}} \phi(Z) {z_i}^{|Z|} \prod_{x \in Z} p_x \prod_{x \in V\setminus Z} (1-p_x)\nonumber\\
                 &= \frac{1}{C_{z_i}} \sum_{j=0}^n {z_i}^j \sum_{\substack{Z \subseteq V\\|Z|=j}} \phi(Z) \Pi_V(Z)\nonumber = \frac{1}{C_{z_i}} \sum_{j=0}^n {z_i}^j \ev_j(\phi).
   \end{align*}
  We can then conclude just like in the proof of Lemma~\ref{lem:evs-to-envss}.
\end{proof}

\subsection{\rev{From Expected Values to Expected Scores}}
\label{subsec:escore-to-ev}

In this section we show reductions in the other direction for $\cshapley$ and
$\cbanzhaf$, under additional assumptions on the class $\calF$.

\paragraph*{Shapley score.}
Let us call a class of Boolean functions~$\calF$\!
\intro{reasonable} if the following problem is in $\mathsf{P}$: given as input (a representation of)~$\phi \in
\calF$, compute~$\phi(\emptyset)$. It is clear that all classes mentioned in
this paper are \kl{reasonable} in that sense. Then, under this assumption:

\begin{proposition}
  \label{prp:eshapley-to-ev}
 We have~$\ev(\calF) \tr
 \eshapley(\calF)$
 for any \kl{reasonable} class~$\calF$\! of Boolean functions.
\end{proposition}
\begin{proof}
  For $\phi:2^Z\to \{0,1\}$, it is well known that the
   \emph{efficiency property} holds:

   \(\sum_{x\in Z} \shapley(\phi,Z,x) = \phi(Z) - \phi(\emptyset).\)
 Let then $\phi \in \calF$ over variables~$V$ and probability values~$p_x$ for
 each $x\in V$; our goal is to compute~$\ev(\phi)$. We have:
  \begin{align}
    \sum_{x\in V} \eshapley(\phi,x) &= \sum_{x\in V} \sum_{\substack{Z\subseteq V\\ x\in Z}} \Pi_V(Z) \times \shapley(\phi,Z,x)\nonumber\\
                                    &= \sum_{Z\subseteq V} \sum_{x\in Z} \Pi_V(Z) \times \shapley(\phi,Z,x)\nonumber\\
                                    &= \sum_{Z\subseteq V}  \Pi_V(Z) \big[\phi(Z) - \phi(\emptyset)\big]\nonumber\\
                                    &= \ev(\phi) - \phi(\emptyset).\label{eq:eefficiency}
  \end{align}
  \rev{We note this is what \cite[Axiom 9]{borkotokey2023expected} calls
    \emph{expected efficiency}, with the slight
  difference that they only consider cases in which $\phi(\emptyset) = 0$.}
  We can compute the left-hand size in $\mathsf{P}$ using oracle calls, and we can
  compute~$\phi(\emptyset)$ in $\mathsf{P}$ as well because~$\calF$ is
  \kl{reasonable},
  therefore we can compute~$\ev(\phi)$ in $\mathsf{P}$ indeed. This concludes the
  proof.
\end{proof}

\begin{example}
    \rev{The sum of all expected
    Shapley values in Table~\ref{tab:shap-eshap} is $0.584$; as $\phiex(\emptyset)=0$, this is exactly
  $\ev(\phiex)$ computed in Example~\ref{exa:ev}.}
\end{example}

This implies, for instance,
that $\eshapley(\calF)$ is intractable over arbitrary circuits,
even monotone bipartite 2-DNF formulas~\cite{provan1983complexity}.
Combining with Theorem~\ref{thm:ev-to-escore}, we obtain in particular:
\begin{corollary}
  \label{cor:equiv}
  We have~$\eshapley(\calF) \equivt \ev(\calF)$ for any \kl{reasonable}
  class $\calF$\!
of Boolean functions.
\end{corollary}
Hence, at least with respect to polynomial-time computability, this settles the
complexity of $\eshapley(\calF)$ for such classes.

\paragraph*{Banzhaf score.}
Next, we show a similar result for the Banzhaf value, under a different,
though commonplace, assumption.

\begin{definition}
  A class~$\calF$ is said to be \intro{closed under conditioning} if the
  following problem is in $\mathsf{P}$: given $\phi\in\calF$ over variables~$V$ and~$x\in V$,
  compute a representation in~$\calF$ of $\phi_{+x}$.
We say
  $\calF$ is \intro{closed under conjunctions (resp., disjunctions)
  with fresh variables} if the following is in $\mathsf{P}$: given $\phi\in \calF$ over
  variables $V$ and $x\notin V$, compute a representation in $\calF$ of the
  Boolean function $\phi \land x$ (resp., $\phi \lor x$).
\end{definition}

\begin{proposition}
  \label{prp:ebanzhaf-to-ev}
 We have~$\ev(\calF) \tr
 \ebanzhaf(\calF)$
 for any class~$\calF$ that is \kl{closed under conditioning} and that is also
 \kl{closed under either conjunctions or disjunctions with fresh
 variables}.
\end{proposition}

This implies that $\ebanzhaf$ is intractable, for instance, over monotone
$2$-CNFs or monotone 2-DNFs.

Proposition~\ref{prp:ebanzhaf-to-ev} requires more work than
Proposition~\ref{prp:eshapley-to-ev}: we do it in two steps by
introducing (yet) another intermediate problem.

\medskip
\begin{center}
\fbox{\begin{tabular}{lp{10cm}}
    \small{PROBLEM} : & $\env(\calF)$ \hspace{.5cm}(\emph{Expected Nested Value})
\\
{\small INPUT} : & A Boolean function $\phi \in \calF$ over variables~$V$ and
probability values~$p_x$ for each $x\in V$
\\
{\small OUTPUT} : & The quantity $\env(\phi) \defeq \sum_{Z\subseteq V} \Pi_V(Z)\sum_{E\subseteq Z} \phi(E)$
\end{tabular}}
\end{center}

The next two lemmas then imply Proposition~\ref{prp:ebanzhaf-to-ev}.

\begin{lemmarep}
  \label{lem:ebanzhaf-to-esv}
  We have $\env(\calF) \tr \ebanzhaf(\calF)$ for any $\calF$ \kl{closed under
  conjunctions (resp., disjunctions) with fresh variables}.
\end{lemmarep}
\begin{proofsketch}
  First, for $\phi':2^{V'} \to \{0,1\}$ and $x\in V'$, we prove the equation
  \begin{equation}
    \label{eq:banzhaf-easy}
    \ebanzhaf(\phi',x) = p_x [\env(\phi'_{+x}) - \env(\phi'_{-x})].
  \end{equation}
  Let then $\phi:2^V \to \{0,1\}$ be the Boolean function for which we
  want to compute $\env(\phi)$, and let $x\notin V$ be a fresh variable.
  We use the closure property of $\calF$ to compute a representation of $\phi' \defeq \phi \odot x$,
  with $\odot$ being $\land$ or $\lor$ depending on the closure property.
  We then show
  using Equation~\eqref{eq:banzhaf-easy} that $\env(\phi)$ can be recovered from
  the single oracle call $\ebanzhaf(\phi',x)$, with $V' \defeq V\cup
  \{x\}$,
  with same probability values for $y\in V$ and $p_x = 1$.
\end{proofsketch}
\begin{proof}
  Let us first show that for
  $\phi':2^{V'} \to \{0,1\}$ and probability values $p_y$ for $y\in V'$ and
  $x\in V'$ we have the equation claimed in the proof sketch, restated here:
  \begin{equation*}
    \ebanzhaf(\phi',x) = p_x [\env(\phi'_{+x}) - \env(\phi'_{-x})].
  \end{equation*}

  The derivation is similar to that of
  Lemma~\ref{lem:envss-to-escore}, but simpler. Observe that $\escore_{\cbanzhaf}(\phi,x) = A
  - B$, where
  \[
  A = \sum_{\substack{Z\subseteq V'\\ x\in Z}} \Pi_{V'}(Z) \sum_{E
  \subseteq Z\setminus \{x\}}
     \phi'(E\cup \{x\})\qquad
  B = \sum_{\substack{Z\subseteq V'\\ x\in Z}} \Pi_{V'}(Z) \sum_{E
  \subseteq Z\setminus \{x\}}
     \phi'(E).
   \]
Let us focus on~$A$. Letting~$V'' \defeq V'\setminus \{x\}$, notice that these are
the variables over which~$\phi'_{+x}$ is defined. Letting~$n\defeq |V''|$, we have

\begin{align*}
  A &= \sum_{\substack{Z\subseteq V'\\ x\in Z}} \Pi_{V'}(Z) \sum_{E
  \subseteq Z\setminus \{x\}}
  \phi'_{+x}(E)
  = p_x \sum_{Z\subseteq V''} \Pi_{V''}(Z) \sum_{E \subseteq Z}
   \phi'_{+x}(E)
  = p_x \env(\phi'_{+x}).
\end{align*}
We can do the same for $B$ to obtain
\(B = p_x \env(\phi'_{-x}),\)
hence the equation.

We now prove Lemma~\ref{lem:ebanzhaf-to-esv} in the case that~$\calF$ is
\kl{closed under conjunctions with
fresh variables}. Let then~$\phi:2^V\to \{0,1\}$, and probabilities $p_y$ for
  $y\in V$. We want to compute $\env(\phi)$. Since $\calF$ is \kl{closed under
  conjunctions with fresh variables}, let~$x \notin V$ and compute a
  representation of $\phi' \defeq \phi \land x$ in $\calF$. We call the oracle
  to $\ebanzhaf$ on $\phi'$ with same probabilities for $y\in V$ and with $p_x
  \defeq 1$. By the above equation (with $V' =  V\cup \{x\}$)
  this immediately gives us $\env(\phi)$ and concludes.

  For the case when $\calF$ is \kl{closed under disjunctions} with fresh variables we
do the same but with $\phi' \defeq \phi \lor x$: now by
Equation~\eqref{eq:banzhaf-easy} the oracle call returns $\big[\sum_{Z\subseteq
V} \Pi_V(Z) \sum_{E\subseteq Z} 1\big]-\env(\phi)$, which is equal to
$\big[\sum_{Z\subseteq V} \Pi_V(Z) 2^{|Z|}\big]-\env(\phi)$. We conclude
the proof
by showing
that the first term is equal to $\prod_{y\in V}
(1+p_y)$, which can be computed in polynomial time, hence we can indeed recover
$\env(\phi)$.
  Indeed, let~$n\defeq
  |V|$, and order the variables of~$V$ arbitrarily as $y_1,\ldots,y_n$. For
  $i\in [n]$, define~$V_i\defeq \{y_j \mid 1\leq j \leq i \in [n]\}$ (note that
  $V_0=\emptyset$), and $d_i \defeq \sum_{Z\subseteq V_i} \Pi_{V_i}(Z)
  2^{|Z|}$. Observe that the quantity that we want is $d_{n}$. But it is clear
  that $d_0 = 1$ and that $d_i =
   \sum_{\substack{Z\subseteq V_i\\y_i\notin Z}} \Pi_{V_i}(Z) 2^{|Z|}
  +\sum_{\substack{Z\subseteq V_i\\y_i\in    Z}} \Pi_{V_i}(Z) 2^{|Z|}
  = (1-p_{y_i})d_{n-1}+p_{y_i}\times d_{n-1}\times 2
  = (1+p_{y_i})d_{n-1}$ for $1 \leq i \leq n$,
  which concludes.
\end{proof}

\begin{lemmarep}
  \label{lem:env-to-ev}
  We have $\ev(\calF) \tr \env(\calF)$ for any class $\calF$ that is
  \kl{closed under conditioning}.
\end{lemmarep}
\begin{proofsketch}
  This is again a rather technical proof by polynomial interpolation, in which
  we curiously seem to need the assumption of closure under conditioning.
\end{proofsketch}
\begin{proof}
  Let
  $\phi \in \calF$ over variables~$V$ with $n\defeq |V|$ and probability
  values~$p_x$ for each $x\in V$; we want to compute $\ev(\phi)$. We use
  polynomial interpolation to compute all the values $\ev_j(\phi)$ for
  $j\in [n]$, after which we can simply return $\sum_{j=0}^n \ev_j(\phi) = \ev(\phi)$.

  Without loss of generality, we can
  assume that~$p_x < 1$ for all $x\in V$. Indeed, if there is $x$ such
  that $p_x=1$, we consider $V'=V\setminus\{x\}$ and $\phi'=\phi_{+x}$.
  Then $\ev_j(\phi)=\ev_{j-1}(\phi')$ for any $j\geq 1$ and
  $\ev_0(\phi)=0$. This is indeed without loss of generality because~$\calF$ is
  closed under conditioning, so that $\phi_{+x}$ is in $\calF$.

Let~$M\defeq \max_{x\in V}p_x<1$.
  Let~$z_0,\ldots,z_{n}$ be~$n+1$ distinct rational values in~$(0,1-M)$.
  For~$i\in [n]$ and~$x\in V$, we define this time $p^{z_i}_x \defeq
  \frac{z_i p_x}{1-p_x}$, and define~$\Pi^{z_i}$ and~$\ev^{z_i}(\phi)$ as expected.
  Again, these are all valid probability mappings.
  Define~$C \defeq \prod_{x\in V} (1-p_x)$. We will show
  that we have $\env^{z_i}(\phi) = \frac{1}{C} \sum_{j=0}^n z_i^j \ev_j(\phi)$,
  which allows us to conclude as in the proof of Lemma~\ref{lem:evs-to-envss}.
  Indeed:
  \begin{align*}
    \env^{z_i}(\phi) &= \sum_{Z\subseteq V} \Pi^{z_i}_V(Z) \sum_{E\subseteq Z} \phi(E)\\
                     &= \sum_{E\subseteq V}  \phi(E) \sum_{E\subseteq Z \subseteq V} \Pi^{z_i}_V(Z)
                     = \sum_{E\subseteq V}  \phi(E) \sum_{E\subseteq Z \subseteq V} \prod_{x\in Z} p^{z_i}_x \prod_{x\in V\setminus Z} (1-p^{z_i}_x)\\
                     &= \frac{1}{C} \sum_{E\subseteq V}  \phi(E) \sum_{E\subseteq Z \subseteq V} \prod_{x\in Z} z_i p_x \prod_{x\in V\setminus Z} (1-p_x -z_i p_x)\\
                     &= \frac{1}{C} \sum_{E\subseteq V}  \phi(E) \prod_{x\in E} z_i p_x \prod_{x\in V\setminus E} [(z_i p_x) + (1-p_x -z_i p_x)]\\
                     &= \frac{1}{C} \sum_{E\subseteq V}  \phi(E) \prod_{x\in E} z_i p_x \prod_{x\in V\setminus E} (1-p_x)
                     = \frac{1}{C} \sum_{j=0}^n \sum_{\substack{E\subseteq V\\|E|=j}} z_i^j \Pi_V(E) \phi(E)
                     = \frac{1}{C} \sum_{j=0}^n z_i^j \ev_j(\phi).
                     \qedhere
  \end{align*}
\end{proof}

And thus, combining with Theorem~\ref{thm:ev-to-escore}, we obtain:
\begin{corollary}
  \label{cor:equiv_banzhaf}
We have~$\ebanzhaf(\calF) \equivt \ev(\calF)$ for any class $\calF$
that is \kl{closed under conditioning} and that is also
\kl{closed under either conjunctions or disjunctions with fresh
variables}.
\end{corollary}

We leave as future work a more systematic (tedious) study of when
$\ev(\calF) \tr \escore_c(\calF)$ holds for other \kl{coefficient
functions}.

\section{DD Circuits}
\label{sec:dD}
We now present algorithms to compute expected Shapley-like scores
in polynomial time over deterministic and decomposable Boolean circuits. Since
computing expected values can be done in linear time over such circuits, the
fact that computing expected Shapley-like scores over them is in
$\mathsf{P}$ is already
implied by our main result, Theorem~\ref{thm:ev-to-escore}. We nevertheless present
more direct algorithms for these circuits as they are easier and more
natural to implement than the convoluted chain of reductions with various
oracle calls and matrix inversions from the previous section. We will moreover
use these algorithms in our experimental evaluation in Section~\ref{sec:exp}.
We start by defining what are these circuits.

\paragraph*{Boolean circuits.} Let~$C$ be a Boolean circuit over variables~$V$,
featuring~$\land$,~$\lor$,~$\lnot$, constant $0$- and $1$-gates, and variable
gates (i.e., gates labeled by a variable in~$V$), \rev{with $\land$- and $\lor$-gates having
arbitrary fan-in greater than $1$}. We
write~$\vars(C)\subseteq V$ the set of variables that occur in the circuit.
The size $|C|$ of~$C$ is its number of wires. For a gate~$g$ of~$C$,
we write~$C_g$ the subcircuit rooted at~$g$, and write
$\vars(g)$ its set of variables.  An $\land$-gate~$g$ of~$C$ is
\emph{decomposable} if for every two input gates $g_1\neq g_2$ of~$g$,
$\vars(g_1) \cap \vars(g_2) = \emptyset$.  We call~$C$ \emph{decomposable}
if all~$\land$-gates in it are.  An $\lor$-gate~$g$ of~$C$ is
\emph{deterministic} if the Boolean functions captured by each pair of distinct
input gates of~$g$ are pairwise disjoint; i.e., there is no assignment that
satisfies them both.  We call~$C$ \emph{deterministic} if all~$\lor$-gates in
it are. A \intro{deterministic and decomposable (d-D) Boolean circuit}~\cite{monet2020solving}
is a Boolean circuit that is both deterministic and
decomposable. An~$\lor$-gate~$g$ is \emph{smooth} if for any input~$g'$ of~$g$
we have~$\vars(g) = \vars(g')$, and $C$ is smooth is all its $\lor$-gates are.
We say that $C$ is \intro{tight} if it satisfies the following three conditions:
(1)~$\vars(C) = V$; (2) $C$ is smooth; and (3) every~$\land$ and every~$\lor$
gate of~$C$ has exactly two children.

\begin{example}
  \rev{The formula $\phiex=(A\land a)\lor(C\land c)$ from our running
    example is not a d-D when interpreted as a circuit, since the $\lor$
    gate is not deterministic. An equivalent d-D is for example
  $\lnot\left[\lnot(A\land a)\land\lnot(C\land c)\right]$, interpreted
as a circuit in the natural way.}
\end{example}

The following is folklore.
\begin{lemma}
  \label{lem:tight}
  Given as input a \kl{d-D circuit}~$C$ over variables~$V$, we can compute in
  $O(|C|\times|V|)$ a \kl{d-D circuit}~$C'$ over~$V$ that is equivalent to~$C$ and that is
  \kl{tight}.
\end{lemma}

\paragraph*{General polynomial-time algorithm.}
It is thus enough to explain how to compute expected Shapley-like scores for
\kl{tight} \kl{d-Ds}; let~$C$ be such a circuit on variables~$V$. We start from
Equation~\eqref{eq:envss-to-escore}, restated here for convenience:
\begin{align*}
  \escore_c(C,x) &= p_x \sum_{k=0}^{|V|-1} \sum_{\ell=0}^k c(k+1,\ell) \big(\ennv_{k,\ell}(C_{1}) - \ennv_{k,\ell}(C_{0})\big).
\end{align*}

Here, $C_{1}$ (resp.,~$C_{0}$) is the circuit~$C$ in which we have replaced
every variable gate labeled by~$x$ by a constant $1$-gate (resp., a constant
$0$-gate). It can easily be checked that~$C_0$ and~$C_1$ are \kl{tight}
\kl{d-Ds}
over~$V\setminus \{x\}$. Therefore, it suffices to compute, for an arbitrary
\kl{tight} \kl{d-D circuit}~$C'$, the $\envss$ quantities. To do this, we
crucially need the determinism and decomposability properties. The idea is to compute corresponding quantities for each gate of the circuit, in a bottom-up fashion. This is similar
to what is done in \cite[Theorem 2]{arenas2023complexity} and \cite[Proposition
4.4]{deutch2022computing}, but the expressions we obtain are more involved because we have a
quadratic number of parameters for each gate of the circuit, as opposed to a
linear number of such parameters in these earlier works.
The resulting algorithm for the whole procedure is shown in
Algorithm~\ref{alg:dD}. Intuitively, the values~$\beta^g_{k,\ell}$ correspond
to the~$\ennv_{k,\ell}$-values for the subcircuit of~$C_1$ rooted at gate $g$, the values~$\gamma^g_{k,\ell}$ correspond
to those for $C_0$, and $\delta$~values are intermediate
quantities that we have to compute. Thus:
\begin{toappendix}
  \subsection{Proof of Theorem~\ref{thm:dD}}
\end{toappendix}
\begin{theoremrep}
  \label{thm:dD}
  Let $c$ be a \kl{tractable} \kl{coefficient function}.
  Given a \kl{d-D} circuit $C$ on variables~$V$, probability values~$p_y$ for~$y\in
  V$, and~$x\in V$, Algorithm~\ref{alg:dD} correctly computes $\escore_c(C,x)$
  in polynomial time. Moreover, if we ignore the cost of arithmetic
  operations, it is in
  time $O\left(|C|\times|V|^5+\mathrm{T}_c(|V|)\times|V|^2\right)$ where
  $\mathrm{T}_c(\alpha)$ is the cost of
  computing the coefficient function on inputs $\leq\alpha$.
\end{theoremrep}

\begin{algorithm}
\caption{Expected Shapley-like scores for deterministic and decomposable Boolean circuits}
\label{alg:dD}
\SetKwBlock{Inductionone}{Compute values~$\delta^g_{k}$ for every gate $g$ in $C$ and~$k\in [n']$ by bottom-up induction on~$C$ as follows:\label{line:begin1}}{end\label{line:end1}}
\SetKwBlock{Inductiontwo}{Compute values~$\beta^g_{k,\ell}$ and~$\gamma^g_{k,\ell}$ for every gate $g$ in $C$ and~$k,\ell \in [n']$ by bottom-up induction on~$C$:\label{line:begin2}}{end\label{line:end2}}
\SetKwInOut{Input}{Input}\SetKwInOut{Output}{Output}
\Input{A \kl{d-D} $C$ on variables~$V$, probability values~$p_y$ for~$y\in
  V$, and~$x\in V$.}
\Output{The value $\escore_c(C,x)$}
\BlankLine
\hrulealg
\BlankLine
\BlankLine
Let $n' = |V|-1$ and let $g_\out$ be the output gate of $C$\;
Make $C$ \kl{tight} using Lemma~\ref{lem:tight}, and call it~$C$ again\;
\BlankLine
\Inductionone{
  \uIf{$g$ is a constant gate or a variable gate with $\vars(g)=\{x\}$}{
        $\delta^g_0 \leftarrow 1$ and $\delta^g_k \leftarrow 0$ for $k\geq 1$\;
    }
    \uElseIf{$g$ is a variable gate with~$\vars(g)=\{y\}$ and $y\neq x$}{
        $\delta^g_0 \leftarrow 1-p_y$, $\delta^g_1 \leftarrow p_y$, and
        $\delta^g_k \leftarrow 0$ for $k\geq 2$\;
    }
    \uElseIf{$g$ is a $\lnot$-gate with input gate $g'$}{
      $\delta^g_k \leftarrow \delta^{g'}_k$ for $k\in [n']$\;
    }

    \uElseIf{$g$ is an $\lor$-gate with input gates $g_1,g_2$}{
      $\delta^g_k \leftarrow \delta^{g_1}_k$ for $k\in [n']$\;
    }
    \uElseIf{$g$ is an $\land$-gate with input gates $g_1,g_2$}{
      $\delta^g_k \leftarrow \sum_{k_1 = 0}^{k}\delta^{g_1}_{k_1} \delta^{g_2}_{k-k_1}$ for $k\in [n']$\;
    }
}
\Inductiontwo{
  \uIf{$g$ is a constant $a$-gate ($a\in \{0,1\}$)}{
    $\beta^g_{0,0}, \gamma^g_{0,0} \leftarrow a$, and
        $\beta^g_{k,\ell}, \gamma^g_{k,\ell} \leftarrow 0$ for $(k,\ell) \neq (0,0)$\;
    }
    \uElseIf{$g$ is a variable gate with~$\vars(g)=\{x\}$}{
    $\beta^g_{0,0} \leftarrow 1$, $\gamma^g_{0,0} \leftarrow 0$, and $\beta^g_{k,\ell}, \gamma^g_{k,\ell} \leftarrow 0$ for $(k,\ell) \neq (0,0)$\;
    }
    \uElseIf{$g$ is a variable gate with~$\vars(g)=\{y\}$ and $y\neq x$}{
    $\beta^g_{0,0}, \beta^g_{1,0}, \gamma^g_{0,0}, \gamma^g_{1,0} \leftarrow 0$,
    $\beta^g_{1,1}, \gamma^g_{1,1} \leftarrow p_x$, and
    $\beta^g_{k,\ell}, \gamma^g_{k,\ell} \leftarrow 0$ for all other values of~$k,\ell$\;
    }
    \uElseIf{$g$ is a $\lnot$-gate with input gate $g'$}{
      $\beta^g_{k,\ell} \leftarrow \binom{k}{\ell} \delta^{g}_k - \beta^{g'}_{k,\ell}$ for $k,\ell\in [n']$\;
      $\gamma^g_{k,\ell} \leftarrow \binom{k}{\ell} \delta^{g}_k - \gamma^{g'}_{k,\ell}$ for $k,\ell\in [n']$\;
    }
    \uElseIf{$g$ is an $\lor$-gate with input gates $g_1,g_2$}{
      $\beta^g_{k,\ell} \leftarrow \beta^{g_1}_{k,\ell} + \beta^{g_2}_{k,\ell}$ for $k,\ell\in [n']$\;
      $\gamma^g_{k,\ell} \leftarrow \gamma^{g_1}_{k,\ell} + \gamma^{g_2}_{k,\ell}$ for $k,\ell\in [n']$\;
    }
    \uElseIf{$g$ is an $\land$-gate with input gates $g_1,g_2$}{
      $\beta^g_{k,\ell} \leftarrow \sum_{k_1 = 0}^{k} \sum_{\ell_1 = 0}^{k_1}\beta^{g_1}_{k_1,\ell_1} \times \beta^{g_2}_{k-k_1,\ell-\ell_1}$ for $k,\ell\in [n']$\;\label{line:and1}
      $\gamma^g_{k,\ell} \leftarrow \sum_{k_1 = 0}^{k} \sum_{\ell_1 = 0}^{k_1}\gamma^{g_1}_{k_1,\ell_1} \times \gamma^{g_2}_{k-k_1,\ell-\ell_1}$ for $k,\ell\in [n']$\;\label{line:and2}
    }
}
\Return ${\displaystyle p_x \sum_{k=0}^{n'} \sum_{\ell=0}^k c(k+1,\ell) \big(\beta^{g_\out}_{k,\ell} - \gamma^{g_\out}_{k,\ell}\big)}$\label{line:return1}\;
\end{algorithm}

We can show that the number of bits of numerators and denominators of
the~$\beta$, $\gamma$ and~$\delta$ values is roughly
$O(b\times|V|)$, for $b$ the bound on the number of bits of numerators and
denominators of all $p_y$ values. Therefore to obtain the exact complexity,
without ignoring the time to perform additions and multiplications over such
numbers, one has to add an
$O(b\times|V|)$
multiplicative factor.

\begin{toappendix}
  Proving Theorem~\ref{thm:dD}, as explained in
  Section~\ref{sec:dD}, boils down to showing how we can compute, given a
  \kl{tight} \kl{d-D circuit}, the $\envss$ quantities. We then show:

\begin{proposition}
  \label{prp:compute-envss}
  Given as input a \kl{tight} \kl{d-D circuit} $C'$ on variables~$V'$ and probability values~$p_y$
  for~$y\in V'$ we can compute all the values $\ennv_{k,\ell}(C)$ for~$k,\ell \in
  [|V'|]$ in $O(|C'|\times |V'|^4)$, ignoring the cost of arithmetic
  operations.
\end{proposition}

Recall that this will be instantiated with~$C'=C_0$ and~$C'=C_1$ for the circuits
$C_0$ and~$C_1$ from Section~\ref{sec:dD} (which should not be confused with
circuit~$C$ of that section). Also note that, even though by
Equation~\eqref{eq:envss-to-escore} we only need to compute the values for~$k\geq
\ell$ (since they are zero when~$k>\ell$), we still do as if we wanted to
naively compute them all. This allows us to obtain cleaner expressions, in which the ranges
for the sums are easier to read. Let us define~$n' \defeq |V'|$.

We first explain how to compute an intermediate quantity that will be needed later.

\begin{definition}
  For a gate~$g \in C'$ and integer~$k \in [n']$, define $\delta^g_k \defeq
  \sum_{\substack{Z \subseteq \vars(g)\\|Z|=k}}  \Pi_{\vars(g)}(Z)$. (Note that
  $\delta^g_k=0$ when~$k > \vars(g')$.)
\end{definition}
Notice that~$\delta^g_k$ only depends on the “structure” of the circuit, but
not on its semantics.

\begin{lemma}
  \label{lem:duh}
  We can compute in $O(|C'|\times n'^2)$ all quantities~$\delta^g_k$, ignoring the
  cost of arithmetic operations.
\end{lemma}
\begin{proof}
  We compute them by bottom-up induction on~$C'$.
  \begin{description}
    \item[Constant gates.] Let~$g$ be a constant gate. Then~$\vars(g) =
      \emptyset$, so $\delta^g_k = 0$ for~$k\geq 1$, and $\delta^g_0=1$
      (indeed~$\Pi_{\emptyset}(\emptyset) = 1$ since this is the neutral
      element of multiplication).
    \item[Input gates.] Let~$g$ be an input gate, with variable~$y$.
      Then~$\vars(g) = \{y\}$, so $\delta^g_k = 0$ for~$k\geq 2$, while
      $\delta^g_0 = \Pi_{\vars(g)}(\emptyset) = 1-p_y$ and $\delta^g_1 =
      \Pi_{\vars(g)}(\{y\}) = p_y$.
      \item[Negation gates.] Let~$g$ be a~$\lnot$-gate with input~$g'$. Notice
        that~$\vars(g) = \vars(g')$. So we have~$\delta^g_k = \delta^{g'}_k$ for all~$k \in [n']$
        and we are done since the values $\delta^{g'}_k$ have already been computed inductively.
      \item[Deterministic smooth~$\lor$-gates.] Let~$g$ be a smooth
    deterministic~$\lor$-gate with inputs $g_1,g_2$. Since~$g$ is smooth we
    have~$\vars(g) = \vars(g_1) = \vars(g_2)$. In particular we have $\delta^g_k = \delta^{g_1}_k$ for all~$k \in [n']$
        and we are done.
      \item[Decomposable~$\land$-gates.] Let~$g$ be a decomposable~$\land$-gate with
    inputs~$g_1,g_2$. Notice that~$\vars(g) = \vars(g_1) \cup \vars(g_2)$ with the union
    being disjoint. We can then decompose~$Z$ into a “left” part~$Z_1\subseteq \vars(g_1)$ of size~$k_1\in \{0,\ldots,k\}$
    and a “right” part~$Z_2\subseteq \vars(g_2)$ of size~$k - k_1$.
    We then have:
      \begin{align*}
        \delta^g_k &= \sum_{\substack{Z \subseteq \vars(g)\\|Z|=k}}  \Pi_{\vars(g)}(Z)
                   = \sum_{k_1 =
                     0}^{k}\sum_{\substack{Z_1
                   \subseteq \vars(g_1)\\|Z_1|=k_1}}\sum_{\substack{Z_2 \subseteq
               \vars(g_2)\\|Z_2|=k-k_1}} \Pi_{\vars(g_1)}(Z_1) \Pi_{\vars(g_2)}(Z_2)\\
                   &= \sum_{k_1 =
                     0}^{k}\sum_{\substack{Z_1
                   \subseteq \vars(g_1)\\|Z_1|=k_1}} \Pi_{\vars(g_1)}(Z_1)\sum_{\substack{Z_2 \subseteq
               \vars(g_2)\\|Z_2|=k-k_1}} \Pi_{\vars(g_2)}(Z_2)
                   = \sum_{k_1 =
                     0}^{k}\sum_{\substack{Z_1
                   \subseteq \vars(g_1)\\|Z_1|=k_1}} \Pi_{\vars(g_1)}(Z_1) \delta^{g_2}_{k-k_1}\\
                   &= \sum_{k_1 =
                     0}^{k} \delta^{g_2}_{k-k_1}\sum_{\substack{Z_1
                   \subseteq \vars(g_1)\\|Z_1|=k_1}} \Pi_{\vars(g_1)}(Z_1)
                   = \sum_{k_1 =
                     0}^{k}\delta^{g_1}_{k_1} \delta^{g_2}_{k-k_1},\\
      \end{align*}
      and we are done.
  \end{description}
  The complexity of every step is $O(n')$ except for $\land$-gates where
  the complexity is $O(n'^2)$; each step needs to be repeated for every
  gate of~$C'$, which gives the stated complexity.
  This concludes the proof of Lemma~\ref{lem:duh}.
\end{proof}

We next define~$\envss$-quantities for all gates of the circuit~$C'$.
\begin{definition}
  For a gate~$g \in C'$ and~$k,\ell \in [n']$, define
  \[\alpha^g_{k,\ell} \defeq \sum_{\substack{Z \subseteq \vars(g)\\|Z|=k}} \sum_{\substack{E\subseteq Z\\|E|=\ell}} \Pi_{\vars(g)}(Z) C'_g(E).\]
\end{definition}
If we can show that we can compute all quantities~$\alpha^g_{k,\ell}$ in
the required complexity then we are
done: indeed, we can then take~$g$ to be the output gate of~$C'$, which gives us
the quantities $\ennv_{k,\ell}(C')$ that we wanted. We show just that in the next lemma.

\begin{lemma}
  \label{lem:alphas}
  We can compute in $O(|C'|\times n'^4)$ all the quantities~$\alpha^g_{k,\ell}$.
\end{lemma}
\begin{proof}
This is again done by bottom-up induction on~$C'$.
\begin{description}
  \item[Constant gates.] Let~$g$ be a constant gate. Then~$\vars(g) =
    \emptyset$, so $\alpha^g_{k,\ell}=0$ when~$(k,\ell)\neq (0,0)$, and
    $\alpha^g_{0,0} = 1$ if~$g$ is a constant $1$-gate and $\alpha^g_{0,0} = 0$
    if it is a constant~$0$-gate.
  \item[Input gates.] Let~$g$ be an input gate, with variable~$y$. Then~$\vars(g) =
    \{y\}$, so all values other than $\alpha^g_{0,0}$, $\alpha^g_{1,0}$ and
    $\alpha^g_{1,1}$ are null, and one can easily check that $\alpha^g_{0,0} =
    \alpha^g_{1,0} = 0$ and $\alpha^g_{1,1} = p_y$.
  \item[Negation gates.] Let~$g$ be a~$\lnot$-gate with input~$g'$. Notice
    that~$\vars(g) = \vars(g')$ and that $C'_g(E) = 1- C'_{g'}(E)$ for any~$E\subseteq
    \vars(g)$. We have
    \begin{align*}
      \alpha^g_{k,\ell} &= \sum_{\substack{Z \subseteq \vars(g)\\|Z|=k}} \sum_{\substack{E\subseteq Z\\|E|=\ell}} \Pi_{\vars(g)}(Z) (1- C'_{g'}(E))
                     = \bigg[\binom{k}{\ell} \sum_{\substack{Z \subseteq \vars(g)\\|Z|=k}}  \Pi_{\vars(g)}(Z)\bigg] - \alpha^{g'}_{k,\ell}\\
                     &= \binom{k}{\ell} \delta^{g}_k - \alpha^{g'}_{k,\ell},\\
    \end{align*}
    and we are done thanks to Lemma~\ref{lem:duh}.
  \item[Deterministic smooth~$\lor$-gates.] Let~$g$ be a smooth
    deterministic~$\lor$-gate with inputs $g_1,g_2$. Since~$g$ is smooth we
    have~$\vars(g) = \vars(g_1) = \vars(g_2)$, and since it is deterministic we
    have~$C'_g(E) = C'_{g_1}(E)+C'_{g_1}(E)$ for any~$E\subseteq \vars(g)$. Therefore
    we obtain $\alpha^g_{k,\ell} = \alpha^{g_1}_{k,\ell} + \alpha^{g_2}_{k,\ell}$ and we
    are done.
  \item[Decomposable~$\land$-gates.] Let~$g$ be a decomposable~$\land$-gate
    with inputs~$g_1,g_2$. Notice that~$\vars(g) = \vars(g_1) \cup \vars(g_2)$ with the
    union being disjoint, and that $C'_g(E) = C'_{g_1}(E\cap \vars(g_1)) \times
    C'_{g_2}(E\cap \vars(g_2))$ and~$\Pi_{\vars(g)}(Z) = \Pi_{\vars(g_1)}(Z\cap \vars(g_1))
    \times \Pi_{\vars(g_2)}(Z\cap \vars(g_2))$ for any~$Z,E\subseteq \vars(g)$. We
    decompose the summations over~$Z$ and~$E$ as we did in the proof of
    Lemma~\ref{lem:duh} for~$\land$-gates. For readability we use colors to
    point out which parts of the expressions are modified or moved around.
    \begin{align*}
      \alpha^g_{k,\ell} &= \textcolor{teal}{\sum_{\substack{Z \subseteq
      \vars(g)\\|Z|=k}}} \sum_{\substack{E\subseteq Z\\|E|=\ell}} \Pi_{\vars(g)}(Z)
      C'_g(E)\\
                     &=  \sum_{k_1 =
                     0}^{k}\sum_{\substack{Z_1
                   \subseteq \vars(g_1)\\|Z_1|=k_1}}\sum_{\substack{Z_2 \subseteq
             \vars(g_2)\\|Z_2|=k-k_1}} \textcolor{teal}{\sum_{\substack{E\subseteq Z_1\cup
       Z_2\\|E|=\ell}}} \Pi_{\vars(g_1)}(Z_1)\Pi_{\vars(g_2)}(Z_2) C'_g(E)\\
                     &=  \sum_{k_1 =
                     0}^{k}\sum_{\substack{Z_1
                   \subseteq \vars(g_1)\\|Z_1|=k_1}}\sum_{\substack{Z_2 \subseteq
               \vars(g_2)\\|Z_2|=k-k_1}} \textcolor{teal}{\sum_{\ell_1 =
           0}^{k_1}} \sum_{\substack{E_1\subseteq
             Z_1\\|E_1|=\ell_1}} \sum_{\substack{E_2\subseteq Z_2\\|E_2|=\ell -
           \ell_1}} \Pi_{\vars(g_1)}(Z_1)\Pi_{\vars(g_2)}(Z_2) C'_{g_1}(E_1)
             C'_{g_2}(E_2)\\
                     &=  \sum_{k_1 = 0}^{k}
                     \sum_{\ell_1 = 0}^{k_1}
                     \sum_{\substack{Z_1 \subseteq
                     \vars(g_1)\\|Z_1|=k_1}}\sum_{\substack{Z_2 \subseteq
                 \vars(g_2)\\|Z_2|=k-k_1}}  \textcolor{teal}{\sum_{\substack{E_1\subseteq
             Z_1\\|E_1|=\ell_1}}} \sum_{\substack{E_2\subseteq Z_2\\|E_2|=\ell -
             \ell_1}} \Pi_{\vars(g_1)}(Z_1)\Pi_{\vars(g_2)}(Z_2) C'_{g_1}(E_1)
             C'_{g_2}(E_2)\\
                     &=  \sum_{k_1 = 0}^{k}
                     \sum_{\ell_1 = 0}^{k_1}
                     \sum_{\substack{Z_1 \subseteq
                     \vars(g_1)\\|Z_1|=k_1}}\sum_{\substack{E_1\subseteq
                 Z_1\\|E_1|=\ell_1}} \sum_{\substack{Z_2 \subseteq
                 \vars(g_2)\\|Z_2|=k-k_1}}   \sum_{\substack{E_2\subseteq
             Z_2\\|E_2|=\ell - \ell_1}} \textcolor{teal}{\Pi_{\vars(g_1)}(Z_1)}\Pi_{\vars(g_2)}(Z_2)
             \textcolor{teal}{C'_{g_1}(E_1)} C'_{g_2}(E_2)\\
                     &=  \sum_{k_1 = 0}^{k}
                     \sum_{\ell_1 = 0}^{k_1}
                     \sum_{\substack{Z_1 \subseteq
                     \vars(g_1)\\|Z_1|=k_1}}\sum_{\substack{E_1\subseteq
                 Z_1\\|E_1|=\ell_1}}\Pi_{\vars(g_1)}(Z_1)C'_{g_1}(E_1) \textcolor{teal}{\sum_{\substack{Z_2 \subseteq \vars(g_2)\\|Z_2|=k-k_1}}
                 \sum_{\substack{E_2\subseteq Z_2\\|E_2|=\ell - \ell_1}}\Pi_{\vars(g_2)}(Z_2) C'_{g_2}(E_2)}\\
                     &=  \sum_{k_1 = 0}^{k}
                     \sum_{\ell_1 = 0}^{k_1}
                     \sum_{\substack{Z_1 \subseteq
                     \vars(g_1)\\|Z_1|=k_1}}\sum_{\substack{E_1\subseteq
                 Z_1\\|E_1|=\ell_1}}\Pi_{\vars(g_1)}(Z_1)C'_{g_1}(E_1)
                 \textcolor{teal}{\alpha^{g_2}_{k-k_1,\ell-\ell_1}}\\
                     &=  \sum_{k_1 = 0}^{k}
                     \sum_{\ell_1 =
                     0}^{k_1}\alpha^{g_2}_{k-k_1,\ell-\ell_1}
                     \textcolor{teal}{\sum_{\substack{Z_1 \subseteq
                     \vars(g_1)\\|Z_1|=k_1}}\sum_{\substack{E_1\subseteq
               Z_1\\|E_1|=\ell_1}}\Pi_{\vars(g_1)}(Z_1)C'_{g_1}(E_1)} \\
                     &=  \sum_{k_1 = 0}^{k}
                     \sum_{\ell_1 =
                     0}^{k_1}\alpha^{g_1}_{k_1,\ell_1} \times \alpha^{g_2}_{k-k_1,\ell-\ell_1}
                     .\\
    \end{align*}
    and we are done.
\end{description}
The complexity is given by that of the step for $\lnot$- and
$\land$-gates, which are
the most costly at $O(n'^4)$ (recall that computing $\binom{k}{\ell}$ is in
$O(k\times\ell)$), which we need to multiply by the size of
the circuit.
This concludes the proof of Lemma~\ref{lem:alphas}.
\end{proof}

In Algorithm~\ref{alg:dD}, the $\beta$ values are the $\alpha$ values for $C_1$
and the $\gamma$ values are the $\alpha$ values for~$C_0$, and they are
computed in a single pass over the circuit~$C$ instead of first computing~$C_1$
and $C_0$ and making passes over these two circuits. Therefore,
Algorithm~\ref{alg:dD} is correct. To obtain the final complexity, we
need to add the cost of line~\ref{line:return1}, which is in
$O(n'^2\times\mathrm{T}(n'))$ ignoring the cost of arithmetic operations,
and remember that $|C_1|$ and $|C_0|$ are in $O(|C|\times|V|)$ by
Lemma~\ref{lem:tight}.

What remains to argue is that the number of bits (numerator and denominator) of
all the $\alpha$ and $\delta$ values stays polynomial. But
for~$\alpha^g_{k,\ell}$ for instance we have
\begin{align*}
  \alpha^g_{k,\ell} &\defeq \sum_{\substack{Z \subseteq \vars(g)\\|Z|=k}} \sum_{\substack{E\subseteq Z\\|E|=\ell}} \Pi_{\vars(g)}(Z) C'_g(E)
                    \leq 2^{2|V|} \max_{Z \subseteq \vars(g)} \Pi_{\vars(g)}(Z).
\end{align*}
If the number of bits of all numerators and denominators of all~$p_x$ is
bounded by $b$, then the numerator of $\alpha^g_{k,\ell}$ is bounded by
$2^{2|V|} 2^{b|V|} = 2^{(b+2)|V|}$, so indeed have a polynomial number of bits
for the numerators, and similar reasoning works for denominators and for the
$\delta$ values.
\end{toappendix}

In the case where all probabilities are identical, we can obtain a lower
complexity by reusing techniques from~\cite{deutch2022computing}:

\begin{toappendix}
  \subsection{Proof of Proposition~\ref{prp:equal_prob}}
\end{toappendix}

\begin{propositionrep}
 \label{prp:equal_prob}
 Let $c$ be a \kl{tractable} \kl{coefficient function}. Given a
 \kl{d-D} $C$ on variables~$V$, a unique probability value $p=p_y$
for all~$y\in V$, and~$x\in V$, $\escore_c(C,x)$ can be computed in time
$O\left(|V|^2\times(|C||V|+|V|^2 + \mathrm{T}_c(|V|))\right)$ assuming unit-cost arithmetic.
\end{propositionrep}

\begin{proof}
  To prove Proposition~\ref{prp:equal_prob} we consider a \kl{d-D}
  circuit $C$ over variables~$V$ with $n=|V|$. For a variable~$x\in V$,
\begin{align*}
  \escore_c(C,x) &= \sum_{\substack{Z\subseteq V\\x\in Z}}\Pi_{V}(Z)\times \score_c(C,Z,x) \\
                 &= \sum_{\substack{Z\subseteq V\\x\in
                 Z}}\Pi_{V}(Z)\sum_{E\subseteq Z\setminus\{x\}}
                 c(|Z|,|E|)\times \left[C(E\cup\{x\})-C(E)\right]\\
  &= \sum_{E\subseteq V \setminus \{x\}}\left[\sum_{E\subseteq Z\subseteq
  V\setminus\{x\}}c(|Z|+1,|E|)\times p_x\times\Pi_{V}(Z)\right][C(E\cup\{x\}) - C(E)]\\
  &= \sum_{\ell=0}^{|V|-1}
  \sum_{\substack{E\subseteq V \setminus \{x\}\\|E|=\ell}}
  \sum_{k=\ell}^{|V|-1}
\left[\sum_{\substack{E\subseteq Z\subseteq
V\setminus\{x\}\\|Z|=k}}c(k+1,\ell)\times p^{k+1}(1-p)^{n-k-1}\right][C(E\cup\{x\}) - C(E)]\\
  &= \sum_{\ell=0}^{|V|-1}
  \sum_{k=\ell}^{|V|-1}
\left[\binom{n-1-\ell}{k-\ell}\times c(k+1,\ell)\times p^{k+1}(1-p)^{n-k-1}\right]
  \sum_{\substack{E\subseteq V \setminus \{x\}\\|E|=\ell}}
[C(E\cup\{x\}) - C(E)]\\
  &= \sum_{\ell=0}^{|V|-1}
  \left[\mathsf{\#SAT}_\ell(C_1)-\mathsf{\#SAT}_\ell(C_0)\right]
  \sum_{k=\ell}^{|V|-1}
\left[\binom{n-1-\ell}{k-\ell}\times c(k+1,\ell)\times p^{k+1}(1-p)^{n-k-1}\right]
\end{align*}
where we set $C_1$ and $C_0$ as usual and $\#\mathsf{SAT}_\ell(C')$ is the
number of satisfying valuations of size~$\ell$ of the circuit~$C'$.
Using the techniques of~\cite{deutch2022computing} (in particular,
Lemma~4.5 of this paper), we can show that all the
$\mathsf{\#SAT}\ell(C')$ values for a tight circuit~$C'$ over variables $|V'|$
can be computed in
$O(|C'|\times |V'|^2)$. So, to compute all $\#SAT_\ell(C_0)$, we first need to make it
tight (in $O(|C|\times |V|)$) and then we have a cost of $O(|C|\times
|V|^3)$.

Now, to compute the rest of the sum, we need to compute for every $\ell$
and $k$ a binomial coefficient in $O(n^2)$, a value of the
coefficient function in $O(\mathrm{T}_c(n))$ and perform the other
multiplications in $O(1)$ assuming unit cost arithmetic. We obtain thus
an algorithm in $O\left(|C|\times |V|^3+|V|^2\times (|V|^2+
\mathrm{T}_c(|V|))\right)$.
\end{proof}

\paragraph*{Quadratic-time algorithm for expected Banzhaf score.} For the
expected Shapley value, instantiating Algorithm~\ref{alg:dD} with $c=\cshapley$
seems to be the best that we can do. For the expected Banzhaf value however, we
can design a more efficient algorithm. We start from
Equation~\eqref{eq:banzhaf-easy}, restated here in terms of circuits:
\begin{equation}
  \label{eq:banzhaf-easy-circuits}
  \ebanzhaf(C,x) = p_x [\env(C_1) - \env(C_0)].
\end{equation}
We can show that $\env$ can be computed in
linear time over \kl{tight} \kl{d-D circuits}, thus obtaining a
$O(|C|\times|V|)$ complexity
for $\escore_{\cbanzhaf}$ over arbitrary \kl{d-D circuits} by Lemma~\ref{lem:tight}:

\begin{toappendix}
  \subsection{Proof of Theorem~\ref{thm:banzhaf-dD}}
\end{toappendix}

\begin{theoremrep}
  \label{thm:banzhaf-dD}
Given a \kl{d-D} $C$ on variables~$V$, probability values~$p_y$ for~$y\in V$,
and~$x\in V$, we can compute in time $O(|C|\times|V|)$ (ignoring the cost of
arithmetic operations) the quantity $\escore_{\cbanzhaf}(C,x)$.
\end{theoremrep}

\begin{proof}
    As argued in Section~\ref{sec:dD}, we need only to prove, thanks to
    Equation~\eqref{eq:banzhaf-easy-circuits}, that
    $\env$ can be computed in linear time for
\kl{tight} \kl{d-D circuits}.
Let~$C'$ be a \kl{tight} \kl{d-D} over variables $V'$ and $p_x$
probability values for all~$x\in V'$. We want to compute $\env(\phi) \defeq
\sum_{Z\subseteq V'} \Pi_{V'}(Z)\sum_{E\subseteq Z} C'(E)$.
(Recall that this will be instantiated with $C'=C_1$ and $C'=C_0$) from
Equation~\eqref{eq:banzhaf-easy-circuits}.)
We do this again by
bottom-up induction on the circuit, computing the corresponding quantities for
very gate. Formally, for a gate $g$ of $C'$, define:
\(\alpha^g \defeq \sum_{Z\subseteq \vars(g)} \Pi_{\vars(g)}(Z)\sum_{E\subseteq Z} C'_g(E).\)
Notice that we want $\alpha^g$ for $g$ the output gate of $C'$.
We show next how this can be done.

\begin{description}
 \item[Constant gates.] Let~$g$ be a constant gate. Then~$\vars(g) =
      \emptyset$, so $\alpha^g$ equals $1$ if $g$ is a constant $1$-gate and
      $0$ if it is a constant $0$-gate.
\item[Input gates.] Let~$g$ be an input gate, with variable~$y$.
      Then~$\vars(g) = \{y\}$, so $\alpha^g = p_y$.
    \item[Negation gates.] Let~$g$ be a~$\lnot$-gate with input~$g'$.\\
      Then $C'_g(E) = 1 - C'_{g_1}(E)$,
      therefore $\alpha^g = \big[\sum_{Z\subseteq \vars(g)} \Pi_{\vars(g)}(Z)\sum_{E\subseteq Z} 1\big] - \alpha^{g'}$.
      We have already observed in the proof of Lemma~\ref{lem:ebanzhaf-to-esv} that the first term is equal to
      $\prod_{y\in \vars(g)}(1+p_y)$, therefore we obtain $\alpha^g = \big[\prod_{y\in \vars(g)}(1+p_y)\big] - \alpha^{g'}$.
\item[Deterministic smooth~$\lor$-gates.] Let~$g$ be a smooth
    deterministic~$\lor$-gate with inputs $g_1,g_2$. Since~$g$ is smooth we
    have~$\vars(g) = \vars(g_1) = \vars(g_2)$, and since it is deterministic we have $C'_g(E) = C'_{g_1}(E) + C'_{g_2}(E)$.
    Therefore $\alpha^g = \alpha^{g_1} + \alpha^{g_2}$.
\item[Decomposable~$\land$-gates.] Let~$g$ be a decomposable~$\land$-gate with
  inputs~$g_1,g_2$. We decompose the sum similarly to what we did in the proof of Theorem~\ref{thm:dD}:

  \begin{align*}
  \alpha^g &= \sum_{Z\subseteq \vars(g)} \Pi_{\vars(g)}(Z)\sum_{E\subseteq Z} C'_g(E)\\
   &= \sum_{Z_1\subseteq \vars(g_1)} \sum_{Z_2\subseteq \vars(g_2)} \Pi_{\vars(g)}(Z_1 \cup Z_2)\sum_{E_1\subseteq Z_1}\sum_{E_2\subseteq Z_2} C'_g(E_2 \cup E_2)\\
   &= \sum_{Z_1\subseteq \vars(g_1)} \sum_{Z_2\subseteq \vars(g_2)} \Pi_{\vars(g_1)}(Z_1)\Pi_{\vars(g_2)}(Z_2)\sum_{E_1\subseteq Z_1}\sum_{E_2\subseteq Z_2} C'_{g_1}(E_1)C'_{g_2}(E_2)\\
   &= \sum_{Z_1\subseteq \vars(g_1)} \Pi_{\vars(g_1)}(Z_1) \sum_{E_1\subseteq Z_1} C'_{g_1}(E_1)\sum_{Z_2\subseteq \vars(g_2)} \Pi_{\vars(g_2)}(Z_2)\sum_{E_2\subseteq Z_2} C'_{g_2}(E_2)\\
   &= \sum_{Z_1\subseteq \vars(g_1)} \Pi_{\vars(g_1)}(Z_1) \sum_{E_1\subseteq Z_1} C'_{g_1}(E_1) \alpha^{g_2}\\
   &=  \alpha^{g_2} \sum_{Z_1\subseteq \vars(g_1)} \Pi_{\vars(g_1)}(Z_1) \sum_{E_1\subseteq Z_1} C'_{g_1}(E_1)\\
   &=  \alpha^{g_2} \alpha^{g_1}.
  \end{align*}
\end{description}
This concludes the proof, as all of this can be done in
$O(|C'|)$,
ignoring the cost of arithmetic operations (and, in any case, the number of bits
stays polynomial).
For it to be true for negation gates, we need to compute
$\prod_{y\in\vars(g)}(1+p_y)$ for every gate, which can also be done during
the bottom-up processing of the circuit.
\end{proof}

\paragraph*{Comparing to and recovering the algorithms of \cite{deutch2022computing} and \cite{abramovich2023banzhaf}.}
We end this section by discussing how this relates to the algorithms
proposed in~\cite{deutch2022computing} and~\cite{abramovich2023banzhaf},
respectively for Shapley and Banzhaf value computation in a deterministic
(non-probabilistic) setting.

First, we note that we can specialize Algorithm~\ref{alg:dD} to the
computation of Shapley values by setting all $p_y$ to~$1$, which means that,
when computing $\env_{k,\ell}$-values, we
only need to consider the case where $k=n$ as all $Z\subseteq V$ with $|Z|<n$
have $\Pi_V(Z)=0$. This leads to the following simplifications:
$\delta_k^g$ values need not be computed as they are all
  $0$'s except for $\delta_{|\vars(g)|}^g=1$;
similarly, $\beta^g_{k,\ell}$ and $\gamma^g_{k,\ell}$ values need only be computed
  when $k=|\vars(g)|$, all other being set to~$0$.
This simplifies the computation to remove a factor of $|V|^2$, and we
essentially recover the algorithm described in~\cite{deutch2022computing}.
Note that the final complexity obtained is $O(|C|\times
|V|^3)$, which is better than the complexity from
Proposition~\ref{prp:equal_prob}.

\begin{toappendix}
  \subsection{Complexity in the Case where All Probabilities are~1}

  As discussed at the end of Section~\ref{sec:dD}, when all probabilities
  are set to~$1$, we recover the algorithm of~\cite{deutch2022computing}
  for non-probabilistic Shapley value computation. We briefly discuss its
  precise complexity.

  In that setting, as discussed, we do not need to compute $\delta^g_k$
  values (line~\ref{line:begin1}--\ref{line:end1}) so we only need to discuss the cost of computing
  $\beta^g_{k,\ell}$ and~$\gamma^g_{k,\ell}$
  (lines~\ref{line:begin2}--\ref{line:end2}) on the one hand, and of
  line~\ref{line:return1} on the other hand. Further, recall that only the setting
  where $k=|\vars(g)|$ is relevant. This means that the main loop to
  compute $\beta$ and $\gamma$ values is run $|\vars(g)|$ times instead
  of $|\vars(g)|^2$ times, and furthermore, that in the case where $g$ is
  an $\land$-gate, its computation involves a single sum, as we can set
  $k_1$ to be $|\vars(g_1)|$ (and thus $k_2$ to be $|\vars(g_2)|$) on
  lines~$\ref{line:and1}$ and~$\ref{line:and2}$.
  Finally, on line~\ref{line:return1}, similarly, we only have one sum
  operator as $\beta^{g_\mathrm{out}}_{k,\ell}$
  and~$\gamma^{g_\mathrm{out}}_{k,\ell}$ are zero when
  $k\neq|\vars(g_{\mathrm{out}})|$.

  Remember that since the circuit needs to be made tight, its size is
  $O(|C|\times|V|)$. We therefore have for complexity:
  \(O\left(|C|\times|V|\times |V|^2+|V|\times
  \mathrm{T}_{\cshapley}(|V|)\right)=O\left(|C|\times |V|^3+|V|\times
|V|^2\right)=O\left(|C|\times|V|^3\right).\)
\end{toappendix}

Second, we observe that the algorithm underlying
Theorem~\ref{thm:banzhaf-dD} has the same complexity as the exact
algorithm in~\cite{abramovich2023banzhaf} for computing (non-expected) Banzhaf values. We
note that \cite{abramovich2023banzhaf} considers \emph{decomposition trees}
instead of d-D circuits, but any decomposition tree is in fact a d-D circuit in
disguise, since a decomposable OR of the form $A \lor B$ can be rewritten as
$\lnot(\lnot A \land \lnot B)$, with the AND being decomposable.
Their algorithm works in linear time on decomposition trees that are tight (see
their Section 3.1), hence we obtain the same complexity while solving a
seemingly more general problem: we study the \emph{expected} Banzhaf values (which degenerates to the non-expected setting when all probabilities are $1$),
and d-D circuits are more general than decomposition trees as they allow
sharing of subexpressions (i.e., the circuit is a DAG instead of a
tree).

\section{Probabilistic Databases}
\label{sec:pdbs}
\paragraph*{(Probabilistic) databases and queries.}
Let $\Sigma = \{R_1,\ldots,R_n\}$ be a \emph{signature}, consisting of
\emph{relation names} each with their associated \emph{arity}~$\ar(R_i) \in
\mathbb{N}$, and~$\const$ be a set of \emph{constants}.  A \emph{fact} over
$(\Sigma,\const)$ is a term of the form~$R(a_1,\ldots,a_{\ar(R)})$, for~$R \in
\Sigma$ and~$a_i \in \const$. A~\emph{($\Sigma,\const)$-database}~$D$ (or
simply a \emph{database}~$D$) is a finite set of facts over $(\Sigma,\const)$.
We assume familiarity with the most common classes of query languages and refer
the reader to~\cite{abiteboul1995foundations,ABLMP21} for the basic
definitions.  A \emph{Boolean query} is a query~$q$ that takes as input a
database~$D$ and outputs~$q(D) \in \{0,1\}$.  If~$q(\bar{x})$ is a query with
free variables~$\bar{x}$ and~$\bar{t}$ is a tuple of constants of appropriate
length, we denote by~$q[\bar{x}/\bar{t}]$ the Boolean query defined by
$q[\bar{x}/\bar{t}](D) = 1$ if and only if~$\bar{t}$ is in the output
of~$q(\bar{x})$ on~$D$. A \emph{tuple-independent probabilistic database}, or
\emph{TID} for short, consists of a database~$D$ together with probability
values $p_f$ for every fact $f\in D$. For a Boolean query $q$ and TID $\bD =
(D,(p_f)_{f\in D})$, the \emph{probability that $\bD$ satisfies $q$}, written
$\Pr(\bD \models q)$, is defined as $\Pr(\bD \models q) \defeq \sum_{D'
\subseteq D\text{~s.t.~}q(D')=1} \Pr(D')$, where $\Pr(D')$ is $\prod_{f \in D'}
p_f \times \prod_{f \in D \backslash D'} (1 - p_f)$. For a fixed Boolean query
$q$, we denote by $\pqe(q)$ the computational problem that takes as input a TID
$\bD$ and outputs $\Pr(\bD \models q)$.

\begin{table}
  \centering
  \small
  \caption{\rev{A TID with two relations \textmd{\textsf{Students}} and \textmd{\textsf{Grades}}}}
  \makebox[0pt]{\rev{\begin{tabular}{clccc}
        \multicolumn{5}{c}{\textsf{Students}}\\
    \toprule
    \textbf{ID} & \textbf{Name} & \textbf{Age} & \textbf{Prob.} &
  \textbf{Prov.} \\
    \midrule
    01 & Alice & 20 & 0.4 & $A$ \\
    02 & Bob & 21 & 0.3 & $B$ \\
    03 & Charlie & 22 & 0.6 & $C$ \\
    04 & Danny & 25 & 0.8 & $D$ \\
    \bottomrule
\end{tabular}}\qquad
  \rev{\begin{tabular}{cccc}
        \multicolumn{4}{c}{\textsf{Grades}}\\
    \toprule
    \textbf{ID} & \textbf{Grade} & \textbf{Prob.} & \textbf{Prov.} \\
    \midrule
    01 & 86 & 0.5 & $a$ \\
    02 & 80 & 0.2 & $b$ \\
    03 & 92 & 0.8 & $c$ \\
    04 & 99 & 0.9 & $d$ \\
    \bottomrule
\end{tabular}}}
  \label{tab:students_grades}
\end{table}

\paragraph*{(Expected) Shapley-like scores.}
Let~$c:\NN \times \NN \to \QQ$ be a \kl{coefficient function}, $q$ a Boolean query,
$D$ a database and $f\in D$ a fact. Following the literature
\cite{DBLP:conf/icdt/LivshitsBKS20,livshits2021shapley,deutch2022computing}, we
define the \emph{Shapley-like score with \kl{coefficients}~$c$ of $f$ in $D$ with
  respect to~$q$}, or simply \emph{score} when clear,
  by\footnote{We point out that the facts of $D$ are traditionally
    partitioned between \emph{endogenous} and \emph{exogenous}
    facts, but we do not make this distinction in our work. This is to
  simplify the presentation, as usual definitions would extend
in a straightforward manner.}
\(\score_c(q,D,f) \defeq \sum_{E \subseteq D\setminus \{f\}}
    c(|D|,|E|)\times \big[q(E\cup \{f\}) - q(E)\big].
  \)
We denote by $\score_c(q)$ the
corresponding computational problem. Let now $\bD = (D, (p_f)_{f\in D})$ be a
TID and $f\in D$, and define the
\emph{expected score of~$f$ for~$q$ in $\bD$} as:
  \[\escore_c(q,\bD,f) \defeq \sum_{\substack{Z\subseteq D, f\in Z}}
    \left(\Pr(Z)\times \score_c(q,Z,f)\right),
  \]
where in $\score_c(q,Z,f)$ we see~$q$ as a function from~$2^Z$ to~$\{0,1\}$.
Define the problem $\escore_c(q)$ as expected. Note that all of this matches
our definitions of expected values, Shapley-like and expected Shapley-like
scores for Boolean functions.

\begin{example}
  \rev{Consider the TID shown
  in Table~\ref{tab:students_grades}. The probability of each tuple is
shown in the next-to-last column of each table, and a tuple identifier
used as a \emph{provenance token} in the last one. Consider the
Boolean query~$q_\mathrm{ex}$ over these tables expressed in SQL as:
``\texttt{SELECT DISTINCT 1 FROM \textup{\textsf{Students}} s JOIN
  \textup{\textsf{Grades}} g ON s.ID = g.ID
WHERE Age < 23 AND Grade >= 85}''. Its \emph{Boolean
provenance}~\cite{senellart2017provenance} can be computed to be exactly
the DNF formula $\phiex=(A\land a)\lor(C\land c)$ from the running example. Its
probability is thus 0.584, as computed in Example~\ref{exa:ev}. The
expected Shapley value of every tuple contributing to making
$q_\mathrm{ex}$ true is given in Table~\ref{tab:shap-eshap}.
}
\end{example}

As usual, if the query $q(\bar{x})$ has free variables, for a tuple $\bar{t}$
of appropriate length we can define similarly, for $f\in D$, the expected score
of $f$ by using the Boolean query $q[\bar{x}/\bar{t}]$ in the above definition.
This score then represents the contribution of $f$ to $\bar{t}$
potentially being in the
query result (one might in particular be interested in explaining why a tuple is
\emph{not} in the query result).

Then Theorem~\ref{thm:ev-to-escore} directly translates into this setting of
Shapley-like scores of facts over probabilistic databases:
\begin{theorem}
  \label{thm:pqe-to-escore}
We have~$\escore_c(q) \tr \pqe(q)$
for any \kl{tractable} \kl{coefficient function}~$c$ and any Boolean query $q$.
\end{theorem}
\begin{proof}
  It suffices to instantiate Theorem~\ref{thm:ev-to-escore} with the set
  of Boolean functions $\calF_q \defeq \{\phi_{q,D} \mid D \text{ is a database}\}$, where
  $\phi_{q,D}: 2^D \to \{0,1\}$, the \emph{Boolean
  provenance}~\cite{senellart2017provenance} of query~$q$ over $D$, is defined by $\phi_{q,D}(D') \defeq q(D')$ for $D'
  \subseteq D$. Here, it is implicit that the Boolean function $\phi_{q,D}$ is represented by $D$ itself.
\end{proof}
In the case of the Shapley value, we can even get a full equivalence from
Corollary~\ref{cor:equiv}:
\begin{corollary}
  \label{cor:pqe-escore}
  We have~$\escore_{\cshapley}(q) \equivt \pqe(q)$
  for any Boolean query $q$.
\end{corollary}
\begin{proof}
  One direction is Theorem~\ref{thm:pqe-to-escore}. The other direction comes
  from Corollary~\ref{cor:equiv}, using the same $\calF_q$ as in the proof of
  Theorem~\ref{thm:pqe-to-escore}, and noticing that $\calF_q$ is
  \kl{reasonable} because the query is fixed.
\end{proof}
In particular, this gives a dichotomy on $\escore_{\cshapley}(q)$
between~$\mathsf{P}$ and
$\mathsf{\#P}$-hard for unions of conjunctive queries (UCQs), or more generally for queries
that are \emph{closed under
homomorphisms}~\cite{dalvi2013dichotomy,DBLP:conf/icdt/Amarilli23}. This should
be compared with the corresponding result for (non-expected)
$\score_{\cshapley}(q)$, where a dichotomy is currently only known for
\emph{self-join--free conjunctive
queries}~\cite{livshits2021shapley,kara2023shapley}.

For Banzhaf values, even though $\escore_{\cbanzhaf}(q) \tr
\pqe(q)$ is true for any Boolean query by Theorem~\ref{thm:pqe-to-escore}, it
is not clear how to obtain the other direction from
Proposition~\ref{prp:ebanzhaf-to-ev}: indeed, the class~$\calF_q$ from
above has in general no reason to be closed under conditioning nor under taking
conjunctions or disjunctions with fresh variables. Yet, we mention that~\cite[Proposition 5.6]{livshits2021shapley}
shows a dichotomy for (non-expected) $\score_{\cbanzhaf}(q)$ for
self-join--free CQs: the tractable queries are the \emph{hierarchical} queries,
while for non-hierarchical queries the problem is $\mathsf{\#P}$-hard. This dichotomy then
directly extends to $\escore_{\cbanzhaf}(q)$: the tractable side follows from
our Theorem~\ref{thm:pqe-to-escore} because $\pqe(q)$ is in $\mathsf{P}$ for hierarchical
queries, while the hardness result is inherited by Fact \ref{fact:eshapley-to-shapley} from the hardness of
$\score_{\cbanzhaf}(q)$ shown in~\cite{livshits2021shapley}.

\paragraph*{Provenance computation and compilation.} Unfortunately, not all queries are tractable for
probabilistic query evaluation. When faced with an intractable query, another
approach is to use the \emph{intensional
method}~\cite{DBLP:series/synthesis/2011Suciu}, which is to compute and
compile the
Boolean provenance of the query on the database in a formalism from knowledge
compilation that enjoys tractable computation of expected values, such as
\kl{d-D
circuits}.
When the provenance has been
computed as a \kl{d-D circuit}, we can use the results from Section~\ref{sec:dD} to
compute the expected Shapley-like scores. This is the route that we take in the
next section to compute these scores in practice.

\section{Implementation and Experiments}
\label{sec:exp}
In this section, we experimentally show that the computation of expected
Shapley-like scores is feasible in practice on some realistic queries
over probabilistic databases, despite the $\mathsf{\#P}$-hardness of the
problem in general and the high $O(|C|\times|V|^5)$ upper bound (see Theorem~\ref{thm:dD})
on the complexity of Algorithm~\ref{alg:dD} for \kl{d-Ds}. The objective is
not to provide a comprehensive experimental evaluation but to simply
validate that algorithms presented in this work have reasonable
complexity for practical applications.

\begin{table*}[t]
    \caption{Provenance computation time, knowledge compilation time and
    method, and total Shapley/Banzhaf computation time for all output
    tuples and all facts, in the deterministic
case and for expected values in the probabilistic case. The queries are the same
TPC-H queries used in~\cite{deutch2022computing} (without the LIMIT
operator used in~\cite{deutch2022computing}). All times
reported are in seconds.}
\footnotesize
    \label{tab:times}
    \resizebox{\columnwidth}{!}{
    \begin{tabular}{r*8r}
    \toprule
    \multicolumn{1}{c}{\textbf{TPC-H}}&
    \multicolumn{1}{c}{\textbf{\# Output}}&
    \multicolumn{1}{c}{\textbf{Provenance}}&
    \multicolumn{1}{c}{\textbf{Compilation}}&
    \multicolumn{1}{c}{\textbf{Compilation}}&
    \multicolumn{1}{c}{\textbf{Avg \kl{d-D}}}&
    \multicolumn{2}{c}{\textbf{Shapley time}}&
    \multicolumn{1}{c}{\textbf{Banzhaf time}} \\
    \multicolumn{1}{c}{\textbf{query}}&
    \multicolumn{1}{c}{\textbf{tuples}}&
    \multicolumn{1}{c}{\textbf{time}}&
    \multicolumn{1}{c}{\textbf{time}}&
    \multicolumn{1}{c}{\textbf{method}}&
    \multicolumn{1}{c}{\textbf{\#gates}}&
    \multicolumn{1}{c}{\textbf{Determ.}}&
    \multicolumn{1}{c}{\textbf{Expect.}}&
    \multicolumn{1}{c}{\textbf{Expect.}}\\
        \midrule
        3&11620&$2.125 \pm 0.029$&$1.226 \pm 0.008$&33\% dec., 67\% tree dec.&22&$0.762 \pm 0.005$&$1.758 \pm 0.011$&$0.468 \pm 0.002$\\
5&5&$1.117 \pm 0.022$&$0.044 \pm 0.000$&100\% tree dec.&1115&$0.766 \pm 0.001$&$40.910 \pm 0.447$&$0.191 \pm 0.000$\\
7&4&$1.215 \pm 0.053$&$0.017 \pm 0.000$&100\% tree dec.&750&$0.269 \pm 0.001$&$9.381 \pm 0.020$&$0.085 \pm 0.000$\\
10&1783&$1.229 \pm 0.023$&$0.018 \pm 0.000$&98\% dec., \02\% tree dec.&5&$0.023 \pm 0.000$&$0.037 \pm 0.000$&$0.015 \pm 0.000$\\
11&61&$0.174 \pm 0.023$&$0.001 \pm 0.000$&87\% dec., 13\% tree dec.&7&$0.001 \pm 0.000$&$0.002 \pm 0.000$&$0.001 \pm 0.000$\\
16&466&$0.247 \pm 0.027$&$0.084 \pm 0.000$&100\% tree dec.&65&$0.159 \pm 0.001$&$0.455 \pm 0.005$&$0.094 \pm 0.001$\\
18&91159&$2.711 \pm 0.298$&$0.749 \pm 0.005$&97\% dec., \03\% tree dec.&4&$0.655 \pm 0.002$&$1.008 \pm 0.007$&$0.490 \pm 0.003$\\
19&56&$1.223 \pm 0.239$&$0.000 \pm 0.000$&100\% dec.&3&$0.000 \pm 0.000$&$0.000 \pm 0.000$&$0.000 \pm 0.000$\\
\bottomrule

    \end{tabular}
    }
\end{table*}

\paragraph*{Implementation.} We rely on, and have extended,
ProvSQL~\cite{DBLP:journals/pvldb/SenellartJMR18}\footnote{\url{https://github.com/PierreSenellart/provsql}}, an
open-source
PostgreSQL extension that computes (between other things) the
Boolean provenance of a query over a database.
We let ProvSQL compute the Boolean provenance
of SQL queries over relational databases as a Boolean circuit,
and have extended this system to add the following features:
\begin{inparaenum}
    \item We compile Boolean provenance into a \kl{d-D} in the simple but
        common \emph{decomposable} case where every $\land$- or
        $\lor$-gate~$g$ is decomposable, i.e., for every two inputs $g_1$
        and $g_2$ to $g$, $\vars(g_1)\cap\vars(g_2)=\emptyset$. Note that, as we
        have already observed in Section~\ref{sec:dD},
        a decomposable $\lor$-gate of the form $A \lor B$ can be
        rewritten, using De Morgan's laws, into a decomposable $\land$-gate.
    \item For cases where this is not possible, we attempt to compile
        Boolean provenance into a \kl{d-D} by computing, if possible, a \emph{tree
        decomposition} of the circuit of treewidth $\leq 10$, and by then
        following the construction detailed in
        \cite[Section 5.1]{amarilli2020connecting} to turn any Boolean
        circuit into a \kl{d-D} in linear time when the treewidth is fixed.
    \item Otherwise, we default to ProvSQL's default compilation of
        circuits into \kl{d-Ds}, which amounts to coding the circuit as
        a CNF using the Tseitin transformation~\cite{tseitin1968complexity}
        and then calling an external knowledge compiler,
        d4~\cite{DBLP:conf/ijcai/LagniezM17}.
    \item We have implemented \emph{directly within ProvSQL}
        Algorithm~\ref{alg:dD} to compute expected Shapley values on
        \kl{d-Ds}, its
        simplification when all $p_y$ are set to~$1$ detailed at the end
        of Section~\ref{sec:dD}, as well as the algorithm to compute
        expected Banzhaf values in the proof of
        Theorem~\ref{thm:banzhaf-dD}. They are all implemented with
        floating-point numbers.
\end{inparaenum}

In particular, this approach is not restricted to queries
on the tractable side of the dichotomy of~\cite{dalvi2013dichotomy}.
In addition, we benefit from the fact that, since late 2021, ProvSQL
stores the provenance circuit in main memory, which speeds
provenance computation up (earlier versions stored the provenance circuit
within the database, on disk).

\paragraph*{Experiment setup.}
Following~\cite{deutch2022computing,abramovich2023banzhaf}, we
used the TPC-H 1~GB benchmark, with standard generated data and 8
standard TPC-H queries
adapted to remove nesting and aggregation, as provided by the authors
of~\cite{deutch2022computing}\footnote{\url{https://github.com/navefr/ShapleyForDbFacts},
archived as \href{https://archive.softwareheritage.org/swh:1:rev:7a46a9fa381194097b81a7aae705d396872b26e3;origin=https://github.com/navefr/ShapleyForDbFacts;visit=swh:1:snp:9de62ad34dffeb67429da0cb0def3183e09c79ad}{swh:1:rev:7a46a9fa381194097b81a7aae705d396872b26e3}}. We use the exact same
queries as in~\cite{deutch2022computing}, except that the \texttt{LIMIT}
operator used for the experiments of that paper was removed, to obtain a
larger and more realistic benchmark (we end up with 105\,154 output tuples for these 8
queries). Probabilities were
drawn uniformly at random for all facts.

Experiments were run on a desktop Linux PC with Xeon W3550 2.80GHz CPU,
64~GB RAM (8~GB of which were made available for PostgreSQL's
shared buffers), running version 14.9 of PostgreSQL; our code has been
incorporated in ProvSQL, and we ran the latest
version of ProvSQL as of December
2023\footnote{Archived on Software Heritage as \href{https://archive.softwareheritage.org/swh:1:rev:bba98e1d96af4c5a25ac672a2f67ea44ed869f8f;origin=https://github.com/PierreSenellart/provsql;visit=swh:1:snp:7e00476360f1bcff9b46b10193b7bf29ebf7d4b0}{swh:1:rev:bba98e1d96af4c5a25ac672a2f67ea44ed869f8f}}.
Data for PostgreSQL is stored on
standard magnetic hard drives in RAID~1.

\paragraph*{Results}
We show in Table~\ref{tab:times} results of these experiments. For each
query, we report: the number of output tuples; the total time required by ProvSQL to evaluate the query and
compute the provenance representation of every output tuple; the total time required by the
compilation of the Boolean provenance circuits of all query results to
\kl{d-Ds}; the method used to produce these \kl{d-Ds}\footnote{A different method might be
used for each output tuple; we report the proportion of ``dec.'' when the
obtained circuit was already decomposable; and of ``tree dec.'' when we
used the
tree decomposition approach; none of the circuits
produced
required using an external knowledge compiler.}; the average number of
gates of the resulting \kl{d-Ds}; the total time needed to compute
(expected) Shapley values of all query outputs for all relevant facts\footnote{By \emph{relevant facts},
we mean here the facts that appear in the provenance circuits. Indeed, the other facts have a score of zero.} in
the deterministic case (where all probabilities are set to~1) and in the
probabilistic case; the same for expected Banzhaf values in the
probabilistic case\footnote{There is of course virtually no difference between computing deterministic
    and expected Banzhaf values, since the algorithm we use (that of
Theorem~\ref{thm:banzhaf-dD}), is the same.}.
All times are
in seconds, repeated over 20 runs of each query, with the mean and
standard deviation shown. To avoid caching provenance across multiple
runs or multiple queries, the ProvSQL circuit was reset each
time and PostgreSQL restarted.

We have the following observations regarding the experimental results
(also compare with the results from Table~1
of~\cite{deutch2022computing}).
\begin{asparaenum}
\item ProvSQL is able to compute in a reasonable amount of time (at most
a couple of seconds) the
output of all queries, along with their Boolean provenance as a circuit;
this contrasts with the results of~\cite{deutch2022computing} where
provenance computation time could take up to 6 hours for query~3, even
when limited to output only~100 tuples; we assume this is the result of
recent ProvSQL optimizations and in-memory storage of the Boolean provenance
circuit.
\item Compilation to a \kl{d-D} takes a time that is comparable
to provenance computation, and uses a combination of
    interpreting the circuit as a decomposable one
    and the tree decomposition algorithm
    of~\cite{amarilli2020connecting}; using an external
    knowledge compiler, which was what was done
    in~\cite{deutch2022computing}, is never required. Note that
    compilation is much faster than reported
    in~\cite{deutch2022computing} (remember that the times
    in~\cite{deutch2022computing} need to be
    multiplied by the number of output tuples, whereas we report the sum
    of all compilation times); the provenance circuits for queries 5 and
    7 could not even be compiled to a \kl{d-D} in~\cite{deutch2022computing}.
\item The total time required for deterministic Shapley value computation
    is of similar magnitude as query evaluation as well and comparable to those reported
    in~\cite{deutch2022computing} (except in one case, for query~19,
    where numbers reported in~\cite{deutch2022computing} are abnormally
    high).
\item Computing expected Shapley values in a probabilistic setting using
    Algorithm~\ref{alg:dD} incurs a higher cost, but remains more
    practical in practice than the high theoretical complexity of this
    algorithm may suggest -- the maximum time required is for query~5,
    with 41 seconds to compute the expected Shapley values over five
    \kl{d-Ds}
    whose average size is over a thousand of gates.
\item
    The total time required
    for Banzhaf value computation is significantly lower than that of deterministic
    Shapley value computation, especially for circuits with large numbers
    of gates or variables, as consistent with the established
    complexities.
\item Though query evaluation and provenance computation can be marred
    with significant performance differences from one run to the next
    (due to disk caching, interaction with PostgreSQL, etc.), there is
    little variation of the performance of knowledge compilation and
    expected Shapley-like value
    computation.
\end{asparaenum}

To summarize, the experiments validate the practicality of the algorithms
presented in this paper for computation of expected Shapley-like scores
over probabilistic databases.

\section{Related work}
\label{sec:relwork}


\rev{\paragraph*{Probabilistic scores from Game Theory.}
While Shapley value~\cite{shapley1953value} is one of the most popular
measures of player contribution in a game, it does not take into account the
probabilistic aspect of some of these games. To address this issue, a long line
of work
\cite{owen1972multilinear,weber1988probabilistic,laruelle2008potential,carreras2015multinomial,carreras2015coalitional,koster2017prediction,borkotokey2023expected}
defines and studies various ways in which scoring mechanisms can be
incorporated in the setting of probabilistic games. For instance, the expected
Shapley value that we study in this paper has been defined in
\cite{borkotokey2023expected} (see Definition 9), where it is shown that, just like the non-probabilistic Shapley value, it is
the only scoring mechanism that satisfies a natural set of axioms over probabilistic games (see Theorems
3 and 4 of~\cite{borkotokey2023expected}). In these works, authors mainly focus on analysing the
characteristics of such scoring mechanisms, in particular in terms
of axioms they verify. To the best of our knowledge, no complexity results or algorithms
are proposed, whereas this is precisely what we study in this paper.}

\paragraph*{Shapley-like scores.}
The authors of
\cite{DBLP:conf/icdt/LivshitsBKS20,reshef2020impact,livshits2021shapley,deutch2022computing}
study the complexity of the problems $\score_{\cshapley}$ and
$\score_{\cbanzhaf}$, over Boolean functions or, most often, instantiated in
the setting of relational databases. See
\cite{bertossi2023shapley} for a survey for databases. In
particular, \cite{deutch2022computing} shows that, for any query, the problem
$\score_{\cshapley}$ can be reduced in polynomial time to probabilistic query
evaluation for the same query. The authors of
\cite{kara2023shapley} show an analogous result for Boolean functions, also
obtaining the other direction of the reduction (under some assumptions).
Formally, define the \emph{model counting problem} for class $\calF$ of Boolean functions as follows:
given as input $\phi \in \calF$ over variables~$V$, compute $\#\phi \defeq \{Z \in 2^V \mid \phi(Z)=1\}$.
They then show (we refer to their article for the definition of closure under OR-substitutions):
\begin{theorem}[{\cite[Corollary 7]{kara2023shapley}}]
We have $\shapley(\calF) \equivt \mc(\calF)$ for any class $\calF$ that is
closed under OR-substitutions.
\end{theorem}
Notice the resemblance between, on the one side, these last two results that we
mention, and on the other side our Corollaries~\ref{cor:equiv}
and~\ref{cor:pqe-escore}. The difference is
that we study the \emph{expected} Shapley values. There is a~priori no
reason for the tractable cases to be the same as the non-expected variant:
indeed, the counting (or probabilistic) version of a problem is often much
harder than the decision one --- for instance probabilistic query evaluation is
often intractable for queries for which regular evaluation is easy.\footnote{So
  one could say that expected Shapley scores are to Shapley scores what
probabilistic query evaluation is to non-probabilistic query evaluation.} By
our results, this phenomenon does not occur for $\escore_{\cshapley}$. Since
$\escore_{\cshapley}$ strictly generalizes $\score_{\cshapley}$ (by
Fact~\ref{fact:eshapley-to-shapley}), the reduction from
$\escore_{\cshapley}$ to $\ev$ is more challenging, and indeed
one can check that our polynomial interpolation proofs are more involved than,
say, \cite[Proposition 3.1]{deutch2022computing}. On the other hand, our life
is made easier to prove the other direction of these equivalences. This
explains why we do not need the assumption of closure under
OR-substitutions in Corollary~\ref{cor:equiv}, and this is also what allows us to obtain a \emph{complete}
equivalence to probabilistic query evaluation in Corollary \ref{cor:pqe-escore}, \emph{no matter the
Boolean query}, whereas in the case of non-expected Shapley values this is only
known for self-join--free CQs \rev{(see also \cite{bienvenu2023shapley} which makes
progress on this problem, but does not solve it completely yet)}.
As for algorithms for d-D circuits
\cite{deutch2022computing,abramovich2023banzhaf}, we refer to the end of
Section~\ref{sec:dD}.

\paragraph*{SHAP-score.}
The authors of
\cite{van2022tractability,arenas2023complexity,van21tractability,arenas21tractability}
study the complexity of computing the SHAP-score. In particular
\cite{van2022tractability,van21tractability} show that it is equivalent to
computing expected values, and \cite{arenas21tractability,arenas2023complexity}
propose polynomial-time algorithms for d-D circuits. Thus, the landscape is
similar to what we obtain here. However, there is to the best of our knowledge
no formal connection between the SHAP-score and the expected scores that we
study here: in a nutshell, we compute the expected Shapley value where the game
function is the Boolean function $\phi$, whereas the SHAP score is computing
the Shapley value where the game function is a conditional expectation of
$\phi$. Hence, the two sets of results seem to be independent.

\section{Conclusion}
\label{sec:conclusion}
We proposed the new notion of expected Shapley-like scores for Boolean
functions, proved that computing these scores can always be reduced in
polynomial-time to the well-studied problem of computing expected values, and
that these two problems are often even equivalent (under commonplace
assumptions). We designed algorithms for deterministic and decomposable
Boolean circuits and implemented them in the setting of probabilistic
databases, where our preliminary experimental results show that these scoring
mechanisms could actually be used in practice.
We leave as future work the study of approximation algorithms for this new
notion. In particular, it is known that $\score_{\cshapley}(q)$ has a
\emph{fully polynomial-time randomized scheme}~\cite{jerrum1986random} whenever
$q$ is a UCQ~\cite{bertossi2023shapley}, and one could study whether this stays
true for the probabilistic variant. Still we note that, since the reduction
from Fact~\ref{fact:eshapley-to-shapley} is parsimonious, we inherit the few
hardness results of the non-probabilistic setting, such as those of
\cite{reshef2020impact} for conjunctive queries with negations.

\begin{acks}
  This research is part of the program
\grantnum[https://descartes.cnrsatcreate.cnrs.fr/]{NRF}{DesCartes}
and is
  supported by the
\grantsponsor{NRF}{National Research Foundation, Prime Minister’s Office, Singapore}{https://www.nrf.gov.sg/}
  under its Campus for Research Excellence and Technological Enterprise
(CREATE) program.
  This work was done in part while Mikaël Monet and Pierre Senellart were visiting the Simons Institute for
the Theory of Computing.
\end{acks}

\clearpage

\bibliographystyle{ACM-Reference-Format}
\bibliography{main}

\end{document}